\documentclass[10pt,a4paper]{article}
\usepackage[english]{babel}
\usepackage[latin1]{inputenc}
\usepackage[T1]{fontenc}
\usepackage{graphicx}
\usepackage{amsmath}
\usepackage{amsfonts}
\usepackage{amssymb}
\usepackage{appendix}
\usepackage{color}
\usepackage{enumitem}
\usepackage{mathabx}
\usepackage{mathrsfs}
\usepackage{dsfont}
\usepackage{amsthm}
\usepackage{bbm}
\usepackage{xfrac}
\usepackage{units}
\usepackage[left=3.0cm, right=3.0cm, bottom=4cm]{geometry}
\usepackage{cite}
\graphicspath{{figures/}}

\newcommand{\comment}[1]{}

\providecommand{\captiontitle}[2]{\caption[#1]{\textit{#1.} #2}}
\providecommand{\deriv}[2]{\frac{\mathrm{d} #1}{\mathrm{d}#2}}
\providecommand{\matel}[3]{\langle #1 \!\! \mid \!\! #2 \!\! \mid \!\! #3 \rangle}
\providecommand{\integral}[3]{\int_{#1}^{#2}  \mathrm{d}#3\:}
\providecommand{\iintegral}[6]{\int_{#1}^{#2} \!\! \mathrm{d}#3 \!\! \int_{#4}^{#5} \!\! \mathrm{d}#6\:}
\providecommand{\iiintegral}[9]{\int_{#1}^{#2} \!\! \mathrm{d}#3 \!\! \int_{#4}^{#5} \!\! d#6 \!\! \int_{#7}^{#8} \!\! \mathrm{d}#9\:}
\providecommand{\integraal}[4]{\int_{#1}^{#2}  \mathrm{d}#3 \mathrm{d}#4\:}

\providecommand{\integ}[2]{\int \limits_{#1}  \mathrm{d}#2\:}

\providecommand{\ket}[1]{\mid \! #1 \rangle}
\providecommand{\cumul}[1]{\langle\!\langle #1 \rangle\!\rangle}
\providecommand{\HI}{H_\textsc{I}}
\providecommand{\ev}[1]{\left\langle #1 \right\rangle}
\providecommand{\kl}{k_\textsc{l}}
\providecommand{\kr}{k_\textsc{r}}
\providecommand{\mul}{\mu_\textsc{l}}
\providecommand{\mur}{\mu_\textsc{r}}

\newtheorem{theorem}{Theorem}[section]

\newtheorem{proposition}[theorem]{Proposition}

\newcommand{\tr}{\operatorname{tr}}
\newcommand{\chit}{\widetilde{\chi}}

\title{Statistics of charge transport and unconventional time ordering}
\author{V. Beaud, G.M. Graf, A.V. Lebedev,\\
\small{Theoretische Physik, ETH Zurich, 8093 Zurich, Switzerland}\\
 G.B. Lesovik\\
\small{L.D. Landau Institute for Theoretical Physics RAS, 117940 Moscow, Russia} }

\begin{document}
\maketitle
\begin{abstract}
	The statistics of charge transport across a tunnel junction with energy--dependent scattering is investigated. A model with quadratic dispersion relation is discussed in general and, independently, in the two limiting cases of a large detector and of a linear dispersion relation. The measurement of charge takes place according to various protocols. It is found that, as a rule, the statistics is expressed by means of time--ordered current correlators, differing however from the conventional ones in the ordering prescription. Nevertheless binomial statistics is confirmed in all cases.
\end{abstract}
\section{Introduction}

The statistics of charge transport $\Delta Q$ through a junction is encoded in its generating function
\begin{equation}
\label{eq_1}
	\chi(\lambda)=\ev{e^{i\lambda \Delta Q}}
	= \sum \limits_{n=0}^{\infty} \frac{(i\lambda)^{n}}{n!}\,\ev{(\Delta Q)^{n}}\,.
\end{equation}
Different notions of charge transport, reflecting different measurement protocols, have been proposed. Each of them has been more or less closely associated to current correlators, which involve some time ordering prescription (like the Dyson \cite{levitov:93} or Keldysh \cite{beenakker:01, belzig:01, lesovik:03, salo:06} variants) or do not \cite{levitov:92}. The resulting statistics differ, but are always expressed in terms of the transparency $T$ of the junction. Somewhat unprecisely, the 3rd cumulant is
 \begin{equation}
\label{eq_2}
\cumul{(\Delta Q)^{3}}\,\propto\, -2T^2(1-T)\,,\qquad
\cumul{(\Delta Q)^{3}}\,\propto\, T(1-T)(1-2T)\,,
\end{equation}
depending on whether time ordering is forgone, resp. imposed. Matters are however more subtle, since it was shown \cite{lesovik:03,sadovski:11} that, in some model with linear dispersion and depending on further details, the first result (\ref{eq_2}) arises even if the correlators are taken to be time ordered. The issue was clarified in \cite{graf:09}, where it was shown that: (i) as a rule, both variants of time ordering require amendment, giving way to unconventional prescriptions; and (ii) in the particular case of that same model, the new procedure yields the second result (\ref{eq_2}), at least in the Dyson variant.

The purpose of the present work is again twofold: First, to confirm item (i), though from a more general perspective, emphasizing the role of gauge coupling; and, second, to compute cumulants on the basis of the new procedure for models or variants not considered before. In particular, we will consider a model with quadratic dispersion and the cases in which the time ordering prescription reduces to the conventional one. In contrast to previous studies, binomial statistics is confirmed in all regimes.\\

For concreteness we present our results first within the setting of \cite{levitov:96}: A system which is endowed with charge and coupled to a spin $\tfrac{1}{2}$, the purpose of which is to serve as a detector, specifically as a galvanometer. Let the Hamiltonian of the combined system be
\begin{equation}
\label{eq_901}
	H(\lambda\sigma_3 /2) = \begin{pmatrix} H(\lambda /2) & 0 \\ 0 & H(-\lambda /2)\end{pmatrix}\,,
\end{equation}
where the coupling $\lambda$ and the Hamiltonian $H(\lambda)$ will be specified later. For now we note that the spin precesses about the $3$-axis, since $\sigma_3$ is conserved.

Let $P\otimes\rho_\mathrm{i}$ be the initial joint state of the system and of the spin, and set $\langle A \rangle = \tr(P A)$, where $A$ is any operator of the system proper. Then
\begin{equation}
\label{eq_907}
	\chi(\lambda) = \left\langle e^{i H(-\lambda/2)t}e^{-iH(\lambda/2)t}\right\rangle
\end{equation}
is the {\it influence functional} describing the spin state alone at a later time $t$,
\begin{equation}\label{if}
\rho_\mathrm{i}
=\begin{pmatrix}\rho_{++}&\rho_{+-}\\\rho_{-+}&\rho_{--}\end{pmatrix}\longmapsto
\rho_\mathrm{f}=
\begin{pmatrix}\rho_{++}&\rho_{+-}\chi(\lambda)\\
\rho_{-+}\chi(-\lambda)&\rho_{--}\end{pmatrix}\,,
\end{equation}
the representation being again in the eigenbasis of $\sigma_3$. With a grain of salt it may also be identified with the generating function in Eq.~(\ref{eq_1}), see item C3 below for details.

We restate (\ref{eq_907}) as
\begin{equation}
\label{eq_911}
	\chi(\lambda) = \left\langle \overrightarrow{T} \exp \left( i \integral{0}{t}{t'}\HI(-\tfrac{\lambda}{2},t')\right)
		\overleftarrow{T} \exp\left(-i \integral{0}{t}{t'}\HI(\tfrac{\lambda}{2},t')\right)\right\rangle\,,
\end{equation}
where $H$ is the Hamiltonian of the isolated system and
\begin{equation*}
	\HI(\lambda, t)= e^{i H t}(H(\lambda)-H)e^{-i H t}
\end{equation*}
is the Hamiltonian in the interaction picture; moreover $\overleftarrow{T}$, $\overrightarrow{T}$ denote the usual and the reversed time ordering. We recall that the time ordering is supposed to occur inside the integrals once the exponentials are expanded in powers of $\lambda$. Eq.~(\ref{eq_911}) follows  by inserting $1=\exp (-iHt)\exp (iHt)$ in the middle of the expectation (\ref{eq_907}) and by performing a Dyson expansion. \\

Let $Q$ be the charge to the right of the junction. It is considered to be a primitive quantity, corresponding to $Q\otimes 1$ for the combined system, but independent of $\lambda$. By contrast and as implied by charge conservation, the current is then a derived quantity, namely the rate of charge: 
\begin{equation}
\label{eq_910}
	I(\lambda\sigma_3 /2) = \deriv{}{t}Q\otimes 1 = i\left[ H(\lambda\sigma_3 /2), Q\otimes 1 \right]
=\begin{pmatrix} I(\lambda /2) & 0 \\ 0 & I(-\lambda /2)\end{pmatrix}\,,
\end{equation}
where $I(\lambda)=i\left[H(\lambda), Q\right]$.

We shall next present two general types of Hamiltonians $H(\lambda)$ to be used in Eq.~(\ref{eq_901}). They are constructed from the Hamiltonian $H$ and from the charge $Q$. At first our goal is to emphasize structure. Later a more concrete physical meaning will be attached to the construction by means of a series of examples for $H$ and of remarks about $H(\lambda\sigma_3 /2)$.

The general Hamiltonians $H(\lambda)$ are as follows.
\begin{enumerate}
	\item[A1.] \textit{Linear coupling}. The Hamiltonian is
	\begin{equation*}
		H(\lambda) = H - \lambda I\,,
	\end{equation*}
	where $I=i[H,Q]$ is the bare current through the junction, i.e. in absence of coupling. Hence,
	\begin{equation*}
		\HI(\lambda, t)=-\lambda I(t)\,,
	\end{equation*}
where $I(t)$ is the bare current in the interaction picture, and Eq.~(\ref{eq_911}) reads
	\begin{equation}
	\label{eq_915}
		\chi(\lambda)=
		\left\langle \overrightarrow{T} \exp\left(i\tfrac{\lambda}{2}\integral{0}{t}{t'}I(t')\right)
		\overleftarrow{T} \exp\left(i\tfrac{\lambda}{2}\integral{0}{t}{t'}I(t')\right)\right\rangle\,,
	\end{equation}
	which is the announced and well-known relation with current correlators (Keldysh time ordering). This setting is closely related to that of \cite{kindermann:01}.

	\item[A2.] \textit{Gauge coupling}. Gauge transformations are generated by local charges $Q$. The Hamiltonian $H$ transforms as
	\begin{align}
	\label{eq_917}
		H(\lambda)& = e^{i \lambda Q} H e^{-i \lambda Q} \\
	\label{eq_918}
			&= H - \lambda I - i \frac{\lambda^{2}}{2}[Q,I] + O(\lambda^{3})\,,\qquad (\lambda\to 0)\,,
	\end{align}
	where $I=i[H,Q]$. (A specific model illustrating that coupling scheme on a spin will be considered shortly in C1.) Local currents are obtained by varying the gauge,
\begin{equation}
\label{eq_903}
	H'(\lambda)\equiv\deriv{}{\lambda}H(\lambda) =-i\left[H(\lambda), Q\right]= -I(\lambda)\,,
\end{equation}
which results in
\begin{equation*}
		I(\lambda) = e^{i \lambda Q} I e^{-i \lambda Q}\,.
	\end{equation*}
The current parallels the kinematic momentum of a particle in that its value is gauge invariant, but its representation as an operator is not, i.e. $I(\lambda)\neq I$ as a rule. This is reflected in the current correlators, which are now those of
	\begin{equation*}
	\HI(\lambda, t)=-\integral{0}{\lambda}{\lambda'}I(\lambda',t)\,,
	\end{equation*}
showing that Eq.~(\ref{eq_915}) is in need of amendment, and not just by replacing $I(t')$ with $I(\mp\lambda/2, t')$. The central result is that $\chi(\lambda)$ can still be expressed by bare current correlators; however Eq.~(\ref{eq_915}) has to be modified as
	\begin{equation}
	\label{eq_919}
		\chi(\lambda)=
		\left\langle \overrightarrow{T}^{*} \exp\left(i\tfrac{\lambda}{2}\integral{0}{t}{t'}I(t')\right)
		\overleftarrow{T}^{*} \exp\left(i\tfrac{\lambda}{2}\integral{0}{t}{t'}I(t')\right)\right\rangle\,,
	\end{equation}
	where the unconventional ordering $\overleftarrow{T}^{*}$ means that the derivative in
	\begin{equation*}
		I(t) = \deriv{Q(t)}{t}\,,\qquad Q(t) = e^{i H t} Q e^{-i H t}
	\end{equation*}
	has to be taken after the time ordering,
	\begin{equation}
	\label{eq_921}
		\overleftarrow{T}^{*}\left(I(t_1)\cdots I(t_n)\right)
		:= \frac{\partial}{\partial t_{n}}\cdots\frac{\partial}{\partial t_{1}}\,
		\overleftarrow{T}\left(Q(t_1)\cdots Q(t_n)\right)\,,
	\end{equation}
	(Matthews' time ordering \cite{matthews:49}); likewise for $\overrightarrow{T}^{*}$. Eq.~(\ref{eq_919}) follows from (\ref{eq_911}) by the identity
	\begin{equation}
	\label{eq_922}
		\overleftarrow{T} \exp\left(-i\integral{0}{t}{t'}\HI(\lambda,t')\right)
		= \overleftarrow{T}^{*} \exp\left(i\lambda \integral{0}{t}{t'} I(t')\right)\,,
	\end{equation}
	which will be discussed in Sect.~\ref{comp}. In summary: The $T$-ordering in connection with $\HI$ is equivalent to ${T}^{*}$-ordering in connection with $-\lambda I$, see Eqs.~(\ref{eq_911}, \ref{eq_919}).
\end{enumerate}
Item A1 is included in A2 as the special case in which
	\begin{equation}
	\label{eq_916}
		[Q,I]=0\,,
 	\end{equation}
see Eq.~(\ref{eq_918}). Then the $T$ and ${T}^{*}$-orderings of currents may be used interchangeably. Moreover,  the operator of current (\ref{eq_910}) becomes insensitive to the inclusion of spin: $I(\lambda)=I$.

Fairly concrete examples illustrating the general types are provided by a {\it single} particle moving on the line. Depending on its position $x\in\mathbb{R}$, the charge $Q$ to the right of the junction is $1$ or $0$, and is in fact implemented as the multiplication operator by the Heaviside function $\theta(x)$. That function can be replaced quite conveniently by a smooth version thereof, $Q=Q(x)$; the size of the region where they differ may be loosely associated with that of the detector. Through second quantization the examples implicitly describe {\it many} independent particles, too.

\begin{enumerate}
	\item[B1.] \textit{Linear dispersion}. The Hamiltonian $H=p+V(x)$ on $L^{2}(\mathbb{R})$ describes a right moving particle. The current $I=i[H,Q]=Q'(x)$ satisfies $[Q,I]=0$, i.e. (\ref{eq_916}), whence the $T$-ordering suffices as a rule. By combining two copies of such a model, scattering between channels of left and right movers can be included. However Eq. (\ref{eq_916}) fails in the limiting but computationally simple case of $Q(x)=\theta(x)$ and $V(x)$ a point interaction at $x=0$ \cite{graf:09}. As a result, $T^{*}$-ordering is required in this special case.

	\item[B2.] \textit{Quadratic dispersion}. The Hamiltonian is $H=p^{2}+V(x)$. Then
	\begin{gather}
	\label{eq_924}
		I = i[H,Q] = p Q'(x) + Q'(x) p\,,\\
	\label{eq_925}
		[Q,I] = 2i Q'(x)^{2} \neq 0\,,
	\end{gather}
which calls for $T^{*}$-ordering as a rule. However, in the limiting case of an ever smoother transition function we have $Q'^{2}\to 0$ (in any reasonable norm). As a result, $T$-ordering should suffice in that special case, as we will confirm.
\end{enumerate}

A few remarks will now address the role of the spin.

\begin{enumerate}
	\item[C1.] We recall \cite{levitov:94} that the Hamiltonian (\ref{eq_901}, \ref{eq_917}) is physically realized by a spin coupled to the current flowing in a wire, as we presently explain. Let the straight wire run along the 1--axis and let $\vec{x}_0$ be the position of the spin in the 12--plane.
The vector potential due to the spin $\vec{\sigma}/2$ is
\begin{equation*}
\vec{A}=\vec{\nabla}f\wedge \frac{\vec{\sigma}}{2}\,,
\end{equation*}
where $f(\vec{x})=\mu|\vec{x}-\vec{x}_0|^{-1}$. (More general functions $f$ are obtained by smearing the position of the spin.) A particle in the wire couples to the spin through $(\vec{p}- \vec{A})\cdot\vec{e}_1\equiv p-\vec{A}\cdot\vec{e}_1$, where
\begin{equation*}
\vec{A}\cdot\vec{e}_1=(\vec{e}_1\wedge\vec{\nabla}f)\cdot \frac{\vec{\sigma}}{2}
=(\vec{e}_1\wedge\vec{\nabla}f)\cdot\vec{e}_3\frac{\sigma_3}{2}=\frac{\partial f}{\partial x_2}\frac{\sigma_3}{2}\,,
\end{equation*}
since by the stated geometry $\vec{e}_1\wedge\vec{\nabla}f$ lies in the 3--direction. We note that $\partial f/\partial x_2=O(r^{-2})$, $(r=|\vec{x}|\to\infty)$; hence, along the wire, $\partial f/\partial x_2=\lambda Q'(x)$ for some coupling $\lambda$ and a function $Q(x)$ of the kind described above. In summary: As $p$ gets replaced by
\begin{equation*}
(\vec{p}- \vec{A})\cdot\vec{e}_1=p-\lambda Q'(x)\frac{\sigma_3}{2}=e^{i\lambda Q\sigma_3/2}pe^{-i\lambda Q\sigma_3/2}\,,
\end{equation*}
so does $H$ by $H(\lambda\sigma_3 /2)$.
\item[C2.]
In the previous item $\lambda Q'(x)$ arises as a connection. It appears with the replacement $\lambda\to\lambda \sigma_3/2$, by which it acts non-trivially on the spin. As we presently explain, that property provides a geometric mechanism (though different from the physical one) by which the rotation of the spin becomes a counter of the transported charge. The mechanism somehow resembles that of a screw, whose motion rigidly links rotation with translation. More precisely, a state $\psi$ of the combined system obeys parallel transport along the line if
\begin{equation*}
\bigl(p-\lambda Q'(x)\frac{\sigma_3}{2}\bigr)\psi=0\,,
\end{equation*}
or equivalently if
\begin{equation*}
 \deriv{\psi}{x}=i\lambda Q'(x)\frac{\sigma_3}{2}\psi\,.
\end{equation*}
Given that a spin state $\psi$ changes by $\mathrm{d}\psi=-i(\vec{\sigma}\cdot\vec{e}/2)\psi\mathrm{d}\theta$ under a rotation by $\mathrm{d}\theta$ about $\vec{e}$, the condition states that a charge transport $\mathrm{d}Q=Q'(x)\mathrm{d}x$ is linked to a precession $\mathrm{d}\theta=-\lambda \mathrm{d}Q$.

\item[C3.] The conclusion of the previous item can not be reached for the physical evolution of the combined system, as generated by the Hamiltonian (\ref{eq_901}), at least not without further ado. In Bohr's spirit \cite{bohr} and in elaboration of \cite{levitov:94} it is convenient to talk about the apparatus in classical terms, here a classical spin $\vec{s}\in\mathbb{R}^{3}$, ($\vert \vec{s}\vert =1$). The spin components $s_i$ have Poisson brackets $\lbrace s_i, s_j\rbrace = \epsilon_{ijk}s_k$ and the Hamiltonian is $H(\lambda s_3)$. In particular Eq.~(\ref{eq_901}) would be recovered by quantization. The equations of motion are
\begin{equation*}
	\dot{s_i} = \lbrace s_i, H(\lambda s_3)\rbrace = \lambda H'(\lambda s_3) \lbrace s_i, s_3\rbrace\,,
\end{equation*}
or, by (\ref{eq_903}),
\begin{equation*}
	\dot{\vec{s}} = -\lambda I(\lambda s_3) \vec{e}_3\wedge\vec{s}\,.
\end{equation*}
The angle of precession thus is
\begin{equation}
\label{eq_906}
	\theta = -\lambda q\,,
\end{equation}
revealing the charge $q$ that has flowed during a time interval $[0,t]$. In this context $\chi$ can be interpreted as the {\it generating function} of the transported charge:
\begin{equation*}
	\chi(\lambda) = \integ{}{q} \widehat{\chi}(q) e^{i  \lambda q}\,.
\end{equation*}
Indeed, since
\begin{equation*}
\rho_\mathrm{s}(\theta)=
\begin{pmatrix}\rho_{++}&\rho_{+-}e^{-i\theta}\\
\rho_{-+}e^{i\theta}&\rho_{--}\end{pmatrix}\,,
\end{equation*}
is the state $\rho_\mathrm{i}$ of the spin after precession by the angle $\theta$, its final state (\ref{if}) is \cite{levitov:96}
\begin{equation*}
	\rho_\mathrm{f} = \integ{}{q} \widehat{\chi}(q)\rho_s(-\lambda q)\,.
\end{equation*}
In view of (\ref{eq_906}), this is consistent with $\widehat{\chi}(q)$ being the  probability (density) of transport $q$, as claimed. The interpretation is however hampered by the fact that $\widehat{\chi}(q)$ may fail to be positive \cite{kindermann:03}.
\end{enumerate}
We should also mention an earlier approach \cite{levitov:93, muzykanskii:03} to charge transport, which does not explicitly model a detector. It is based on two measurements of the charge $Q$, occurring at times $0$ and $t$. The transported charge is then identified with the difference $\Delta Q$ of their outcomes. The associated generating function is
\begin{equation}
\label{eq_909}
	\chit(\lambda) = \left\langle e^{i \lambda Q(t)}e^{-i \lambda Q}\right\rangle\,,
\end{equation}
at least in the case when the initial state $\rho$ is an eigenstate of $Q$ or an incoherent superposition of such, i.e., for $[\rho,Q]=0$; then $\chit$ actually agrees with the expression (\ref{eq_907}). We will use the definition beyond this restriction, because it is irrelevant in the limit of large times. (See however \cite{shelankov:03} for the unrestricted definition.) By
$(e^{i H t}e^{i \lambda Q}e^{-i H t})e^{-i \lambda Q} = e^{i H t}(e^{i \lambda Q}e^{-i H t}e^{-i \lambda Q})$ we may restate the generating function in a form closer to (\ref{eq_907}),
\begin{equation*}
	\chit(\lambda) = \left\langle e^{i Ht}e^{-iH(\lambda)t}\right\rangle\,,
\end{equation*}
with $H(\lambda)$ as in Eq.~(\ref{eq_917}); and further in terms of current correlators by means of (\ref{eq_922})
\begin{equation}\label{eq_919bis}
	\chit(\lambda) = \left\langle\overleftarrow{T}^{*} \exp\left(i\lambda \integral{0}{t}{t'} I(t')\right)\right\rangle\,,
\end{equation}
where the star can again be dropped under the assumption (\ref{eq_916}).

We shall now describe the main result. It confirms the binomial statistics of charge transport in a variety of situations. Specifically, in the long--time limit the 2nd and 3rd cumulants of charge transport are
\begin{align}
\label{eq_19}
	\lim\limits_{t\to\infty}\frac{1}{t}\cumul{(\Delta Q)^{2}}
	&= \frac{1}{2 \pi} \integral{\mur}{\mul}{E} T(E)\left( 1-T(E)\right)\,,\\
\label{eq_20}
	\lim\limits_{t\to\infty}\frac{1}{t}\cumul{(\Delta Q)^{3}}
	&= \frac{1}{2 \pi} \integral{\mur}{\mul}{E} T(E)\left( 1-T(E)\right)\left( 1-2\,T(E)\right)\,,
\end{align}
where $\mur<\mul$ are the Fermi energies of the states incoming from the right and the left sides of the junction, and $T(E)$ is the transmission probability (transparency) at energy $E$. (Eq.~(\ref{eq_19}) was first obtained in~\cite{lesovik:89} without any time ordering prescription.) The results apply to
\begin{itemize}
\item[D1.] either generating function, Eq.~(\ref{eq_919}) or (\ref{eq_919bis});
\item[D2.] Hamiltonians with linear or quadratic dispersion relation (see items B, but in second quantization);
\item[D3.] independently of how sharp the jump of $Q$ is, i.e. of the width over which $Q(x)$ differs from $\theta(x)$.
\end{itemize}
Of some interest is the way that independence arises. The time ordering (\ref{eq_921}) is spelled out for $n=2$ as
\begin{align}
	\overleftarrow{T}^{*}(I(t_1)I(t_2))
		&= \frac{\partial}{\partial t_{2}}\frac{\partial}{\partial t_{1}}\bigl(Q(t_1)Q(t_2)\theta(t_{1}-t_{2})+Q(t_2)Q(t_1) \theta(t_{2}-t_{1})\bigr)\nonumber\\
		&= \overleftarrow{T}(I(t_1)I(t_2))+[Q(t_1),I(t_1)]\delta(t_{1}-t_{2})\,,\label{ect}
\end{align}
and thus differs from the usual one, $\overleftarrow{T}$, by {\it contact terms} supported at coinciding times; likewise for $n=3$, where
\begin{equation}
\overleftarrow{T}^{*}(I_1 I_2 I_3) = \overleftarrow{T}(I_1 I_2 I_3) + 3\delta(t_2 - t_3) \overleftarrow{T}(I_1 [Q_2, I_2])
	 	+ \delta(t_1 - t_2)\delta(t_2 - t_3)[Q_1,[Q_1,I_1]]\label{ect1}
	\end{equation}
with the shorthand notation $I_i=I(t_i)$ (for general $n$, see \cite{graf:09}). Depending on circumstances, the terms in the expansions contribute variably to the invariable results (\ref{eq_19}, \ref{eq_20}), as we now detail.

Items D2, D3 come with interpolating parameters: The Fermi wavelength $\lambda_\textsc{F}$, with $\lambda_\textsc{F}\to 0$ as the linear dispersion is approached through rescalings
\begin{equation*}
\frac{\lambda_\textsc{F}}{2}(p\pm \lambda_\textsc{F}^{-1})^2-\frac{\lambda_\textsc{F}^{-1}}{2}\to \pm p
\end{equation*}
of right and left movers; and the width $l$ of the transition region. In the limit $l\to\infty$ \cite{Ng} of an ever larger detector $Q$ ceases to be defined. In the opposite limit $l\to 0$ and in the case of quadratic dispersion ($\lambda_\textsc{F}>0$), the commutator $[Q,I]$ diverges even in the sense of distributions, because of $Q'(x)\to \delta(x)$ (see item B2). One can though discuss the limits of the correlators (\ref{eq_19}, \ref{eq_20}) in these limits, and of their parts, but not those of the models themselves. However a model with $\lambda_\textsc{F}=0$, $l=0$ exists \cite{graf:09}, describing a
scatterer and a detector which are both pointlike and coincident.

Let us discuss Eq.~(\ref{eq_919bis}) first. The contributions of contact terms to the cumulants (\ref{eq_19}, \ref{eq_20}) are {\it non-trivial}, except in the limits for which $l/\lambda_\textsc{F}\to \infty$, as shown in Fig.~\ref{fig_1}.
\begin{figure}[hbtp]
\centering
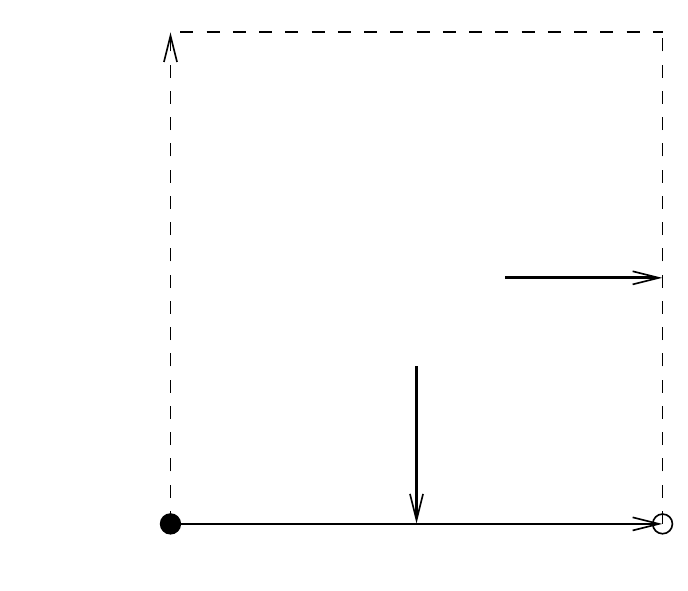
\captiontitle{Matthews' vs. usual time ordering of current correlators}{The parameter range $(l,\lambda_\textsc{F})$ of models is shown as a square, which includes solid parts of the boundary. In the limits (thick arrows) of linear dispersion ($\lambda_F\to0$), or large detectors ($l\to\infty$), the contact terms appearing in Matthews' time ordering, $T^{*}$, vanish. In these limiting cases there is agreement with usual time ordering, $T$.}
\label{fig_1}
\end{figure}
As for Eq.~(\ref{eq_919}) the same is true for the 3rd cumulant; however for the 2nd the contact terms cancel between $\overrightarrow{T}^{*}$ and $\overleftarrow{T}^{*}$.\\

Relation with the literature is eased through books and the review articles on noise and counting statistics, and among them \cite{blanter:00,kindermann:03,sadovski:11}.\\

The plan of the article is as follows. In Sect.~\ref{sec_models} we introduce a model with quadratic dispersion relation and review the main features of a limiting case with linear dispersion. In Sect.~\ref{sec_overview} we will explain the broad structure of the computation of the cumulants and emphasize the methods. The first part of Sect.~\ref{sec_derivations} is devoted to the detailed derivation of asymptotic binomial statistics for the model with quadratic dispersion. In the second and third parts, the limiting cases of a large detector and of linear dispersion are given independent derivations. In Sect.~\ref{comp} we recall the reason for the $T^{*}$ time ordering and discuss the equivalence between the methods used here, resp. elsewhere, as e.g.~in \cite{lesovik:03}. Finally, two appendices collect some auxiliary results. Appendix~\ref{sec_temp_distr} contains the long--time limits of some distributions, while in Appendix~\ref{sec_matrix_elements} matrix elements of charge and current are computed.\\

\section{The models}
\label{sec_models}

\subsection{The quadratic dispersion model}
\label{subsec_quadratic_model}

We consider two conducting leads connected through a junction and model the whole device by independent fermions moving on the real line. The single-particle Hamiltonian is $H=p^{2}+V(x)$, acting on the Hilbert space $L^{2}(\mathbb{R})$. The kinetic energy is quadratic in the momentum $p=-id/dx$; the potential $V$ describes the junction and vanishes away from it, i.e., outside of some interval $[ -x_{0},\,x_{0}]$. The potential will enter the discussion only through its reflection and transmission amplitudes, $r(k)$ and $t(k)$. They can be read off from the \textit{Lippmann--Schwinger (LS) states} $\vert \psi_{k} \rangle$: Continuum eigenstates of $H$ of incoming momentum $k\neq 0$ and eigenvalue $E = k^{2}$ have wave--functions $\psi_{k}$ given outside that interval as
\begin{equation}\label{eq_21}
\begin{aligned}
        k > 0\; : \quad \psi_{k}(x)&=\begin{cases} e^{ikx}+r(k) e^{-ikx}\,, \quad & (x  <  -x_{0})\\
	t(k) e^{ikx}\,, \quad &(x  >  x_{0})
	\end{cases}\\
	k < 0\; : \quad \psi_{k}(x)&=\begin{cases} t(k) e^{ikx}\,, \quad & (x <  -x_{0})\\
	e^{ikx}+r(k) e^{-ikx}\,, \quad &(x  >  x_{0})\,.	
	\end{cases}
\end{aligned}
\end{equation}
Note that states with $k>0$ ($k<0$) have incoming parts that are right (left) moving. By the Schr\"odinger equation the scattering matrix
\begin{equation*}
	S(k) = \begin{pmatrix}
			t(k) & r(-k)\\
			r(k) & t(-k)
		\end{pmatrix}
\end{equation*}
is unitary. In particular, the transmission and reflection probabilities are even in $k$, whence $T(k^2) := \vert t(\pm k) \vert^2$ and $R(k^2) := \vert r(\pm k) \vert^2$, and satisfy
\begin{align}
\label{eq_23}
	T(E) + R(E) = 1\,.
\end{align}
LS states form a (continuum) basis
of $L^{2}(\mathbb{R})$ normalized as
\begin{equation}
\label{eq_24}
	\frac{1}{2 \pi} \integral{\mathbb{R}}{}{k} \vert \psi_{k} \rangle \langle \psi_{k} \vert = \mathds{1}\,.
\end{equation}
Time--reversal invariance of $H$ is, incidentally, a property which is not relied upon, in that the above discussion still applies when $p$ is replaced by $p-A(x)$, at least as long as $A$ has the same support properties as $V$.

As mentioned in the introduction, the charge to the right of the junction may be implemented on $L^{2}(\mathbb{R})$ as a multiplication operator, $Q=Q(x)$. More specifically, we assume
\begin{equation}
\label{eq_25}
	Q(x) = \begin{cases} 0\,, \qquad &(x < x_{0})\\1\,, &(x \gg x_{0})\,.
	\end{cases}
\end{equation}
The left and right leads are assumed to be reservoirs with energy levels occupied up to Fermi energies $\mul,\mur>0$ biased by $V=\mul-\mur>0$. The occupation of LS states thus is
\begin{equation}
\label{eq_27}
 \rho(k)= \begin{cases}1\,,\quad &(-\kr\le k\le \kl)\\
0\,,\quad &\text{otherwise}\end{cases}\,,
\end{equation}
where $k_{\textsc{l},\textsc{r}}=(\mu_{\textsc{l},\textsc{r}})^{1/2}$; or for short $\rho(k)=\theta(k\in J)$ where $J:=[-\kr,\kl]$. More precisely, $\rho\equiv\rho(k)$ is the single--particle density matrix ($0\le \rho=\rho^*\le 1$) of the many--particle state $\langle\cdot\rangle$; actually $\rho$ determines a {\it quasi--free fermionic state} $\langle\cdot\rangle$, which for practical purposes means that expectations of many--particle operators can be computed by means of Wick's rule. As a matter of fact, for $\rho$ a projection as in (\ref{eq_27}) the state $\langle\cdot\rangle$ is necessarily quasi--free.

\subsection{The linear dispersion model}
\label{subsec_linear_model}

We briefly review the main features of the linear dispersion relation model used in \cite{graf:09} and underlying the computations of Sect.~\ref{subsec_limit_linear}. For a more detailed exposition, we refer to the original paper.

In the limit of long times (or low frequencies) it appears appropriate to linearize the dispersion relation near the Fermi energy. A suitable model arises by reinterpreting the two leads on either side of the junction: Rather than viewing them as non--chiral half--lines, they are now (full) chiral lines. In absence of scattering, which now amounts to a cut junction, the Hamiltonian is linear in the momentum $p=-id/dx$ and is given as
\begin{equation*}
	H_0 = \begin{pmatrix} p &0 \\ 0&p \end{pmatrix}\,,
\end{equation*}
on $L^{2}(\mathbb{R})\oplus L^{2}(\mathbb{R})$. A point scatterer is then placed at $x=0$; it results in a unitary scattering matrix
\begin{equation*}
	S = \begin{pmatrix}
			r & t'\\
			t & r'
		\end{pmatrix}\,,
\end{equation*}
which is independendent of energy. In particular, $T=|t|^2=|t'|^2$ and $R=|r|^2=|r'|^2$ still satisfy Eq.~(\ref{eq_23}). The single-particle charge operator is the projection onto the right lead,
\begin{equation*}
	Q = \begin{pmatrix} 0 &0 \\ 0 &1 \end{pmatrix}\,,
\end{equation*}
and the initial single--particle density matrix is the projection
\begin{equation*}
	\rho = \begin{pmatrix} \theta(\mul-p) &0 \\ 0 &\theta(\mur-p) \end{pmatrix}\,,
\end{equation*}
representing two infinitely deep Fermi seas biased by $V = \mul - \mur>0$. The condition $[\rho,Q]=0$, underlying the unrestricted use of the generating function (\ref{eq_909}), is satisfied here.

A feature of the model is that the scattering process is instantaneous in the sense that the position of the point scatterer coincides with that of the detector. As a result, $[Q,I]\neq 0$ (\cite{graf:09}, Eq.~(3.17)), and the contact terms arising from $T^{*}$-ordering matter. In terms of the discussion given at the end of the introduction, the model has vanishing length scale $l=0$.

However the scattering process can be regarded as strictly causal by separating the two positions by $l>0$. This is achieved by replacing the charge operator $Q$ by its regularization $Q_l:= Q\theta(\vert x\vert > l)$, and accordingly the current $I$ by $I_l:=i[H,Q_l]=Q[\delta(x-l)-\delta(x+l)]$. Then the commutator $[Q_l,I_l]$ vanishes and with it all the contact terms.

\section{Overview}
\label{sec_overview}

Before engaging in the detailed computation of the cumulants (\ref{eq_19}, \ref{eq_20}) it is worthwhile giving an overview of the methods involved, and illustrating them in simple instances. The physical setting has been discussed at length in the introduction and will be recalled only briefly. We consider two leads separated by a tunnel junction, with particles in an initial multi--particle state $\langle\cdot\rangle$. We investigate the statistics of charge transport, $\Delta Q$, across the junction and during a time $t$. Specifically, we are interested in its moments $\ev{(\Delta Q)^{n}}$, determined as the expansion coefficients of some generating function, see Eq.~(\ref{eq_1}); and actually in the long--time limit of the associated cumulants $\cumul{(\Delta Q)^n}$.\\

\noindent{\bf Generating functions.} In the introduction two distinct generating functions were presented:
\begin{align}
\label{eq_402}
	\chit(\lambda) &= \ev{\overleftarrow{T}^{*} \exp\left(i\lambda \integral{0}{t}{t'} I(t')\right)}\,,\\
	\chi(\lambda)&=
		\ev{\overrightarrow{T}^{*} \exp\left(i\tfrac{\lambda}{2}\integral{0}{t}{t'}I(t')\right)
		\overleftarrow{T}^{*} \exp\left(i\tfrac{\lambda}{2}\integral{0}{t}{t'}I(t')\right)}\,,\nonumber
\end{align}
where
\begin{equation}
I(t) = e^{i H t} I e^{-i H t}
\label{eq_403bis}
\end{equation}
is the current across the junction. It is expressed in terms of the charge $Q(t)$ to its right as $I(t) =dQ(t)/dt$, whence
\begin{equation}
\label{eq_403ter}
	I = i[H,Q]\,.
\end{equation}
We shall refer to $\chit$ and $\chi$ as the generating functions of
the \textit{first} and of the \textit{second kind}, respectively.\\

\noindent{\bf Results.} The 2nd and 3rd cumulants exhibit asymptotic binomial behavior,
\begin{align}
\label{eq_413}
	\lim\limits_{t\to\infty}\frac{1}{t}\cumul{(\Delta Q)^{2}}
	&= \frac{1}{2 \pi} \integral{\mur}{\mul}{E} T(E)\left( 1-T(E)\right)\,,\\
\label{eq_414}
	\lim\limits_{t\to\infty}\frac{1}{t}\cumul{(\Delta Q)^{3}}
	&= \frac{1}{2 \pi} \integral{\mur}{\mul}{E} T(E)\left( 1-T(E)\right)\left( 1-2\,T(E)\right)\,,
\end{align}
where $T(E)$ is the transparency and $\mu_{\textsc{l},\textsc{r}}$ are the Fermi energies on the left and right leads, in various instances and for either generating function. Specifically:
\begin{itemize}
	\item[-] in the quadratic dispersion model, contact terms matter up to the special case of an ever smoother step of the charge operator $Q(x)$ (see Section~\ref{subsec_quadratic_model}).
	\item[-] in the linear dispersion model, contact terms vanish, except in the special case of instantaneous scattering (see Section~\ref{subsec_linear_model}).
\end{itemize}

In the rest of this section we address the methods used to obtain the results from the generating functions. \\

\noindent{\bf $T^{*}$-ordering.} It is convenient to recall the expansion in contact terms for $\overleftarrow{T}^{*}$-ordered products. With the shorthand notation $A_i=A(t_i)$, $A\equiv I,Q$, Eqs.~(\ref{ect}, \ref{ect1}) read
\begin{align}
\label{eq_405}
	\overleftarrow{T}^{*}(I_1 I_2)
		&= \overleftarrow{T}(I_1 I_2)+\delta(t_1-t_2)[Q_1, I_1]\,,\\
\label{eq_406}
	\overleftarrow{T}^{*}(I_1 I_2 I_3) &= \overleftarrow{T}(I_1 I_2 I_3) + 3\delta(t_2 - t_3) \overleftarrow{T}(I_1 [Q_2, I_2])
	 	+ \delta(t_1 - t_2)\delta(t_2 - t_3)[Q_1,[Q_1,I_1]]\,.
\end{align}
The ordering by $\overrightarrow{T}^{*}$ yields the same expansions, up to a minus sign for contact terms involving an odd number of commutators. A general expression for products of all orders may be found in \cite{graf:09}.\\

\noindent{\bf Cumulants.} Based on the generating function $\chit$ of the first kind we have
\begin{align}
\label{eq_409}
	\cumul{(\Delta Q)^2} &= \integral{0}{t}{^{2}t}\cumul{\overleftarrow{T}^{*}(I_1 I_2)}
		= \integral{0}{t}{^{2}t}\cumul{\overleftarrow{T}(I_1 I_2)} + \integral{0}{t}{t_1} \cumul{[Q_1,I_1]}\,,\\
	\cumul{(\Delta Q)^3} &= \integral{0}{t}{^{3}t}\cumul{\overleftarrow{T}^{*}(I_1 I_2 I_3)}\nonumber\\
\label{eq_410}
		&= \integral{0}{t}{^{3}t}\cumul{\overleftarrow{T}(I_1 I_2 I_3)}
			+3\integral{0}{t}{^{2}t}\cumul{\overleftarrow{T}(I_1[Q_2,I_2])}+\integral{0}{t}{t_1}\cumul{[Q_1,[Q_1,I_1]]}\,,
\end{align}
where $d^nt=dt_1\ldots dt_n$. At first, similar equations are obtained for the moments by means of Eqs.~(\ref{eq_402}) and (\ref{eq_405}, \ref{eq_406}). Moments can then be replaced by cumulants; indeed, their combinatorial relation is universal, and hence the same on both sides of the equations.

Based on the generating function $\chi$ of the second kind we likewise find
\begin{align}
\label{eq_411}
	\cumul{(\Delta Q)^2} &= \frac{1}{4}\integral{0}{t}{^{2}t}[\cumul{\overrightarrow{T}^{*}(I_1 I_2)} + 2\,\cumul{I_1 I_2}
			+  \cumul{\overleftarrow{T}^{*}(I_1 I_2)}]
		= \integral{0}{t}{^{2}t} \cumul{I_1 I_2}\,\\
\label{eq_412}
	\cumul{(\Delta Q)^3} &= \frac{1}{8}\integral{0}{t}{^{3}t}[\cumul{\overrightarrow{T}^{*}(I_1 I_2 I_3)}
			+ 3\, \cumul{\overrightarrow{T}^{*}(I_1 I_2)I_3} + 3\, \cumul{I_1\overleftarrow{T}^{*}(I_2 I_3)}
			+ \cumul{\overleftarrow{T}^{*}(I_1 I_2 I_3)}]\,,
\end{align}
where the 3rd cumulant may also be expanded using (\ref{eq_405}, \ref{eq_406}) for both time arrows. We note that the cumulants of the second kind involve $\overleftarrow{T}^{*}$-ordered current correlators already present in those of the first kind, and more. However, due to the symmetry between usual and reversed time orderings, the contact terms in the 2nd cumulant (\ref{eq_411}) mutually cancel.\\

\noindent{\bf Wick's rule.} The many--particle state $\langle\cdot\rangle$ is the quasi--free state\cite{lundberg:76} determined by the single--particle density matrix $\rho$. Let $\widehat{A}$ be the second quantization of the single--particle operator $A$. Correlators of second quantized operators can be reduced to the level of first quantization thanks to \textit{Wick's rule}. In particular, with $\rho' := 1-\rho$ we have
\begin{align}
\label{eq_414a}
	\cumul{\widehat{A}} &= \langle\widehat{A}\rangle = 0\,, \\
\label{eq_415}
	\cumul{\widehat{A}\widehat{B}} &= \tr(\rho A \rho' B)\,, \\
\label{eq_416}
	\cumul{\widehat{A}\widehat{B}\widehat{C}} &= \tr(\rho A \rho' B \rho' C) - \tr(\rho A \rho' C \rho B)\,.
\end{align}
The expressions follow from the usual formulation of the rule which involves creation and annihilation operators $\psi^*(a)$, $\psi(b)$ of single--particle states $a$, $b$: Expectations of products of such are computed by way of complete contraction schemes and reduced to just two kinds of contractions, $\langle\psi^*(a)\psi(b)\rangle=\langle b|\rho|a\rangle$, $\langle\psi(b)\psi^*(a)\rangle=\langle b|\rho'|a\rangle$, with further ones vanishing. The second quantization $A\mapsto \widehat{A}$ is defined for rank--one operators $A=|a_1\rangle\langle a_2|$ as
\begin{equation*}
\widehat{A}=\psi^*(a_1)\psi(a_2)-\langle a_1|\rho|a_2\rangle\,
\end{equation*}
and then extended by linearity in $A$. We stress the ``zero-point subtraction'' of $\langle a_1|\rho|a_2\rangle=\tr(\rho A)$, which implies $\langle\widehat{A}\rangle = 0$, but drops out from higher cumulants. The l.h.s. of Eq.~(\ref{eq_415}) gives rise to just one non--vanishing connected contraction scheme and thus equals
\begin{equation*}
\cumul{\psi^*(a_1)\psi(a_2)\psi^*(b_1)\psi(b_2)}= \langle\psi^*(a_1)\psi(b_2)\rangle\langle\psi(a_2)\psi^*(b_1)\rangle=\langle b_2|\rho|a_1\rangle\langle a_2|\rho'|b_1\rangle\,,
\end{equation*}
in agreement with the r.h.s..

The charge and current operators mentioned earlier in this section are meant in second quantization. We will henceforth denote them by $\widehat{Q}$ and $\widehat{I}$. \\

\noindent{\bf GNS space and Schwinger terms.} The reader may skip this item. It is in fact about some fine points which remain without practical consequences. Strictly, $\psi^*(a)$, $\psi(b)$ act on the \textit{GNS space} \cite{lundberg:76} of $\langle\cdot\rangle$ but, as we explain below, one may sometimes pretend it is replaced by Fock space. We only consider the case when $\rho$ is a projection. An operator $A$ admits a second quantization $\widehat{A}$ if $B=[\rho, A]$ is Hilbert-Schmidt, i.e. $\tr(B^{*}B)<\infty$. The traces (\ref{eq_415}, \ref{eq_416}) exist if $A$ and the other observables satisfy that condition. The (first quantized) operators $Q$ and $I$ do so in both models of Sect.~\ref{sec_models}. However, $\tr(\rho I)$ is well-defined only in the model with quadratic dispersion, and $\tr(\rho Q)$ in neither.

For $\rho=0$ the stated condition becomes trivial, and the GNS space is the Fock space. If its operators $\widehat{A}$ are used on another quasi-free state $\rho$, then (\ref{eq_415}, \ref{eq_416}) are still valid if the traces exist, but (\ref{eq_414a}) is to be replaced by $\cumul{\widehat{A}} = \langle\widehat{A}\rangle =\tr(\rho A)$. The difference is in the ``zero-point subtraction'' which may diverge, even for $[\rho, A]$ Hilbert-Schmidt. The point of the GNS space is that $\widehat{A}$ still remains defined there.

 On the GNS space we have \cite{lundberg:76}
\begin{equation*}
[\widehat{A}, \widehat{B}] =\widehat{[A,B]} + S(A,B)\cdot 1\,,\qquad S(A,B)= \tr(\rho A \rho' B ) -\tr(\rho B \rho' A )\,,
\label{Sch}
\end{equation*}
where the last term is known as a {\it Schwinger term}. We infer
\begin{gather}\label{Sch1}
\widehat{A}(t)=\widehat{A(t)}+i\int_0^t\mathrm{d}t'\,S(H,A(t'))1\,,\\
\label{Sch2}
\cumul{[\widehat{A}, \widehat{B}]}=S(A,B)\,.
\end{gather}
Upon pretending that the GNS space is just Fock space, we have $[\widehat{A}, \widehat{B}] =\widehat{[A,B]}$ and $\widehat{A}(t)=\widehat{A(t)}$. However, Eq.~(\ref{Sch2}) still holds true, because (\ref{eq_415}) does. The same conclusion is obtained from $\cumul{[\widehat{A}, \widehat{B}]}=\cumul{\widehat{[A,B]}}=\tr(\rho[A,B])$.

Let us comment on the significance of Schwinger terms for the cumulants (\ref{eq_409}, \ref{eq_410}), where $I_i$ is now to be read as $\widehat{I}(t_i)$ (and likewise for $Q_i$). First, it may be replaced by $\widehat{I(t_i)}$, because the difference seen in (\ref{Sch1}) drops out from the results. Second, the contact terms $\cumul{[\widehat{Q(t_1)}, \widehat{I(t_1)}]}$ and $\cumul{[\widehat{Q(t_1)}, \widehat{[Q(t_1),I(t_1)]}]}$ are properly Schwinger terms. Informally however they may be understood in the context of Fock space, as discussed. \\

\noindent{\bf A simple case.} We illustrate the methods by considering a simple example: The 2nd cumulant of the second kind, Eq.~(\ref{eq_411}), for the model with quadratic dispersion and with state (\ref{eq_27}). Traces may be evaluated using the basis (\ref{eq_24}) of LS states. We so obtain
\begin{equation}
	\cumul{\widehat{I}_1\widehat{I}_2} = \tr(\rho I_1 \rho' I_2)
		= \frac{1}{(2\pi)^2} \integ{\mathbb{R}^{2}}{^{2}k} \matel{1}{\rho I_1}{2}\matel{2}{\rho' I_2}{1}
\label{eq_418}
		= \frac{1}{(2\pi)^2} \integ{J\times\mathbb{R}\setminus{J}}{^{2}k} e^{i(E_1-E_2)(t_1-t_2)}\vert\matel{1}{I}{2}\vert^2\,,
\end{equation}
with the shorthand notations $\ket{i}=\ket{\psi_{k_i}}$, $E_i=k_i^2$, and $J=[-\kr, \kl]$. The last equality follows from Eq.~(\ref{eq_403bis}) and the eigenvalue equation $H\ket{i}=E_i\ket{i}$. \\

\noindent{\bf Time integrals.} The long-time limit of (\ref{eq_411}) now calls for
\begin{equation*}
	\lim\limits_{t\to\infty}\frac{1}{t}\integral{0}{t}{^{2}t} e^{i(E_1-E_2)(t_1-t_2)} = 2\pi \delta(E_1-E_2)
		= \frac{2\pi}{2\vert k_1\vert}(\delta(k_1-k_2) + \delta(k_1+k_2))\,.
\end{equation*}
The first equality is by Eq.~(\ref{eq_37}) below and the second by $E_i = k_i^2$. Only one of the diagonals $k_1=\pm k_2$ openly intersects the integration domain $J\times\mathbb{R}\setminus{J}$ in Eq.~(\ref{eq_418}), and in fact just for $k_1=-k_2\in [\kr,\kl]$, see Fig.~\ref{fig_2}. Hence
\begin{figure}[hbtp]
\centering
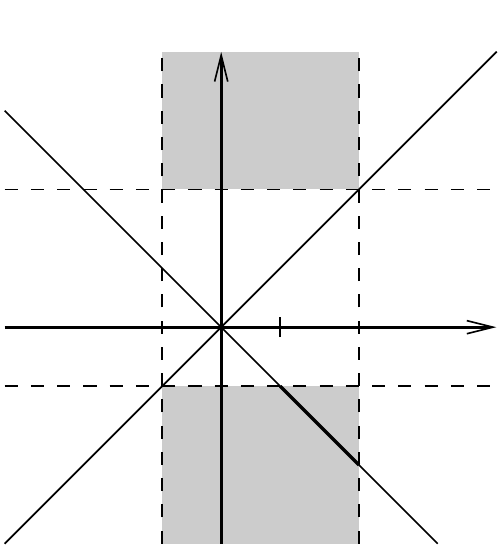
\captiontitle{Integration over $k_1$, $k_2$ in Eq.~(\ref{eq_418})}{The integration domain $J\times\mathbb{R}\setminus{J}$ (shaded) and its intersection with $E_1=E_2$ (diagonals).}
\label{fig_2}
\end{figure}

\begin{equation}
\label{eq_421}
	\lim\limits_{t\to\infty}\frac{1}{t}\cumul{(\Delta Q)^2}
		=\frac{1}{2\pi}\integral{\kr}{\kl}{k_1}\frac{\vert\matel{1}{I}{-1}\vert^2}{2k_1}\,,
\end{equation}
where $\ket{-1}:=\ket{\psi_{-k_1}}$. Appendix~\ref{sec_temp_distr} collects further long--time limits of distributions.\\

\noindent{\bf Matrix elements of current.} The current $I$ is given in Eq.~(\ref{eq_924}). We compute the relevant matrix element for $k_1>0$ and observe that Eq.~(\ref{eq_21}) distinguishes between the cases $\pm k>0$; however only the expressions for $x>x_0$ matter here because of the support properties of $Q'$ seen in (\ref{eq_25}):
\begin{align}
	\matel{1}{I}{-1} &= \matel{1}{pQ'(x)+Q'(x)p}{-1}\nonumber \\
		&= \integral{-\infty}{\infty}{x} \overline{t(k_1)} e^{-ik_1 x}(pQ'(x)+Q'(x)p)(e^{-ik_1 x}+r(-k_1)e^{ik_1 x})\nonumber\\
		&= 2k_1\overline{t(k_1)}r(-k_1)\integral{-\infty}{\infty}{x}Q'(x)= 2k_1\overline{t(k_1)}r(-k_1)\,,
\label{eq_422}
\end{align}
where the third equality is by partial integration. Hence
\begin{equation}
\label{eq_423}
	\vert\matel{1}{I}{-1}\vert^2 = (2k_1)^2 T(k_1^2 ) R(k_1^2 ) = (2k_1)^2 T(k_1^2 )(1-T(k_1^2 ))\,,
\end{equation}
and Eq.~(\ref{eq_413}) follows by substituting $E=k_1^2$ in Eq. (\ref{eq_421}). In the simple case considered we thus confirm the binomial statistics.

Further computations of matrix elements of current may be found in Appendix \ref{sec_matrix_elements}. In particular, we mention
\begin{equation}
\label{eq_424}
	\matel{1}{I}{1} = 2k_1 T(k_1)\,,\qquad(k_1>0)\,.
\end{equation}

\noindent{\bf Matrix elements of charge.} The cumulant (\ref{eq_411}) we just computed is the simplest among (\ref{eq_409}-\ref{eq_412}) in that it does not involve the charge operator $Q$. In preparation of the other cases, it pays to look at the relation between matrix elements of $Q$ and of $I$, still though within the model of quadratic dispersion. In view of the support property (\ref{eq_25}) of $Q(x)$, computing its Fourier transform demands a regularization at $x\to+\infty$:
\begin{equation}
\label{eq_425}
	\widehat{Q}(k) = \lim\limits_{\varepsilon\downarrow 0}\integral{-\infty}{\infty}{x}Q(x)e^{-i(k-i\varepsilon)x}
		=(-i)\lim\limits_{\epsilon\downarrow 0}\frac{\widehat{Q'}(k)}{k-i\,\varepsilon}
		= (-i)\frac{\widehat{Q'}(k)}{k-i\,0}\,.
\end{equation}
The result is a distribution in $k$, and so is $\matel{1}{Q}{2}$ in $k_1,k_2$; whereas $\matel{1}{I}{2}$ is a smooth function of these variables. By (\ref{eq_403ter}) we have the equation
\begin{equation}
\label{eq_425a}
\matel{1}{I}{2}=i(E_1-E_2)\matel{1}{Q}{2}\,,
\end{equation}
which however can not be uniquely solved for $\matel{1}{Q}{2}$, because the distributional equation $xF(x)=0$ admits the non--trivial solutions $F(x)\propto\delta(x)$. In fact, in view of the Sokhatsky-Weierstrass (SW) formula
\begin{equation}\label{ws}
\frac{1}{x-i\,0}-\frac{1}{x+i\,0}=2\pi i \delta(x)\,,
\end{equation}
the general solution is
\begin{equation}
	\label{eq_427}
		i\matel{1}{Q}{2} = \frac{\matel{1}{I}{2}^{(+)}}{E_1-E_2+i\,0} + \frac{\matel{1}{I}{2}^{(-)}}
		{E_1-E_2-i\,0}\,,
\end{equation}
where $\matel{1}{I}{2} = \matel{1}{I}{2}^{(+)}+\matel{1}{I}{2}^{(-)}$ is any split of the matrix element of current. It takes $\matel{1}{Q}{2}$ to make it unique, at least up to terms vanishing for $E_1=E_2$, which may still be shifted between the two terms $\matel{1}{I}{2}^{(\pm)}$.

In summary: An expression for $\matel{1}{I}{2}$ does not entail one for $\matel{1}{Q}{2}$; rather conversely, including the split. Such expressions will be derived in Appendix~\ref{sec_matrix_elements}. An important case is when $k_1\neq k_2$, whence $E_1=E_2$ arises by $k_1=- k_2$; then $\matel{1}{I}{2}^{(-)}=\matel{1}{I}{2}$ and Eq.~(\ref{eq_427}) simply reads
\begin{equation}
\label{eq_507}
	\matel{1}{Q}{2} = (-i)\frac{\matel{1}{I}{2}}{E_1-E_2-i\,0}\,.
\end{equation}
Another case is
\begin{equation}
\label{eq_431}
	\matel{-1}{I}{-1}^{(-)} = 2k_1 R(k_1)\,,\qquad(k_1>0)\,.
\end{equation}
\comment{
\begin{itemize}
\item $k_1,k_2>0$. We have
\begin{gather*}
	i\matel{1}{Q}{2} = i\overline{t(k_1)}t(k_2)\widehat{Q}(k_1-k_2)
		= \overline{t(k_1)}t(k_2)\frac{\widehat{Q'}(k_1-k_2)}{k_1-k_2-i\,0}\,,\\
(k_1+k_2)(k_1-k_2-i\,0)=E_1-E_2-i\,0\,,
\end{gather*}
and hence $\matel{1}{I}{2}^{(-)}=\matel{1}{I}{2}$, $\matel{1}{I}{2}^{(+)}=0$. The same conclusion holds when $k_1$ and $k_2$ have opposite signs, or more generally when $k_1$ and $k_2$ are in disjoint intervals. In fact, this latter case reduces to the previous one, since $E_1-E_2=0$ may only occur when $k_1$ and $k_2$ have opposite signs.
\item $k_1,k_2<0$. A straightforward computation yields
\begin{align*}
	i\matel{1}{Q}{2}
	=\frac{\widehat{Q'}(k_1-k_2)}{k_1-k_2 - i\,0}
		+\overline{r(k_1)}\frac{\widehat{Q'}(-k_1-k_2)}{-k_1-k_2- i\,0}
	 +r(k_2)\frac{\widehat{Q'}(k_1+k_2)}{k_1+k_2- i\,0}+
\overline{r(k_1)}r(k_2)\frac{\widehat{Q'}(k_2-k_1)}{k_2-k_1 - i\,0}\,.
\end{align*}
The two middle terms have non--vanishing denominators in the stated range and may thus be linked to either term $\matel{1}{I}{2}^{(\pm)}$. Using $(k_1+k_2)(k_1-k_2\mp i\,0)=E_1-E_2\pm i\,0$ we may choose to set
\begin{equation}
	\label{eq_430}
	\matel{1}{I}{2}^{(-)} = -(k_1+k_2)\overline{r(k_1)}r(k_2)\widehat{Q'}(k_2-k_1)\,.
\end{equation}
\end{itemize}
}

\section{Derivations}
\label{sec_derivations}

\subsection{The quadratic dispersion model}
\label{subsec_derivations_quadratic}

Using the methods introduced in the previous section we shall derive the asymptotic binomial distribution (\ref{eq_413}, \ref{eq_414}) for both generating functions, $\chit$ and $\chi$. We shall do so first for the model with quadratic dispersion relation of Section~\ref{subsec_quadratic_model}. It will become evident that contact terms are crucial. In other words Matthews' time--ordering can not be replaced by ordinary time--ordering, except in limiting cases, if the correct result is to be found. Two such cases, namely that of a large detector and of a linear dispersion, will be given independent treatments in the following sections.\\

\noindent{\bf 2nd cumulant of the first kind.} The cumulant is given in Eq.~(\ref{eq_409}) as
\begin{equation*}
	\cumul{(\Delta Q)^2} = \integral{0}{t}{^{2}t}\cumul{\overleftarrow{T}^{*}(\widehat{I}_1 \widehat{I}_2)}
		= \mathrm{A} + \mathrm{B}
\end{equation*}
with
\begin{align*}
	\mathrm{A} &= \integral{0}{t}{^{2}t}\cumul{\overleftarrow{T}(\widehat{I}_1 \widehat{I}_2)}=2\iintegral{0}{t}{t_1}{0}{t_1}{t_2} \cumul{\widehat{I}_1 \widehat{I}_2} &\textnormal{(main term)}\,,\\
	\mathrm{B} &= \integral{0}{t}{t_1}\cumul{[\widehat{Q}_1,\widehat{I}_1]} &\textnormal{(contact term)}\,.
\end{align*}
The connected correlators are computed by Wick's rule (\ref{eq_415}) and the resulting traces evaluated in the basis of LS states (\ref{eq_21}). We obtain
\begin{align}
\label{eq_504}
	\mathrm{A} &= \frac{2}{(2\pi)^2} \integ{J\times\mathbb{R}\setminus{J}}{^{2}k} \iintegral{0}{t}{t_1}{0}{t_1}{t_2}
		e^{i(E_1-E_2)(t_1-t_2)}\vert\matel{1}{I}{2}\vert^2\,,\\
\label{eq_505}
	\mathrm{B} &= \frac{t}{(2\pi)^2} \integ{J\times\mathbb{R}\setminus{J}}{^{2}k}
		(\matel{1}{Q}{2}\matel{2}{I}{1}-\matel{1}{I}{2}\matel{2}{Q}{1})\,.
\end{align}
In relation to the overview above we shall next (i) discuss time integrals and (ii) express matrix elements of charge in terms of those of current. We will do likewise later for all cumulants. In the present case the first item concerns only the main term, the second only the contact term.

(i) The asymptotic long--time behavior of the main term is given by Eq.~(\ref{eq_36}) with $x=E_1-E_2$ and the substitution $t_2\mapsto t_1-t_2$:
\begin{equation}
\label{eq_506}
	\frac{1}{t}\iintegral{0}{t}{t_1}{0}{t_1}{t_2} e^{i(E_1-E_2)(t_1-t_2)} \;\xrightarrow[t\to+\infty]{}\;
	\frac{i}{E_1-E_2+i\,0}\,,
\end{equation}
as distributions in $k_1$ and $k_2$. In Eq.~(\ref{eq_504}) $\matel{1}{I}{2}$ then qualifies as a test function, being essentially the Fourier transform of the  compactly supported function $Q'(x)$.

(ii) Within the integration domain (\ref{eq_505}) $\matel{1}{Q}{2}$ is given by Eq.~(\ref{eq_507}), and $\matel{2}{Q}{1}$ is then obtained by exchanging $k_1$ and $k_2$, or by complex conjugation.

Collecting terms we so obtain
\begin{align*}
	\lim\limits_{t\to+\infty}\frac{1}{t}\cumul{(\Delta Q)^2}
		&= \frac{1}{(2\pi)^2} \integ{J\times\mathbb{R}\setminus{J}}{^{2}k}
			(\frac{2i}{E_1-E_2+i\,0}-\frac{i}{E_1-E_2-i\,0}+\frac{i}{E_2-E_1-i\,0})\vert\matel{1}{I}{2}\vert^2\nonumber\\
		&= \frac{1}{2\pi}\integral{k_\textsc{r}}{k_\textsc{l}}{k_1}\frac{\vert\matel{1}{I}{-1}\vert^2}{2k_1}
		= \frac{1}{2\pi}\integral{\mu_\textsc{r}}{\mu_\textsc{l}}{E} T(E)(1-T(E))\,,
\end{align*}
as claimed. The second equality is by the SW formula (\ref{ws}) and the earlier remark restricting $k_1$ to $[\kr,\kl]$ (see Fig.~\ref{fig_2}); the last one by Eq.~(\ref{eq_423}) with $E=k_1^2$.\\

\noindent{\bf 3rd cumulant of the first kind.} Though longer, the computation of the 3rd cumulant retains the same two key ingredients: (i) the evaluation of time integrals and (ii) the expression of matrix elements of charge in terms of those of current. By Eq.~(\ref{eq_410}) we have
\begin{equation*}
	\cumul{(\Delta Q)^3} = \integral{0}{t}{^{3}t}\cumul{\overleftarrow{T}^{*}(\widehat{I}_1 \widehat{I}_2 \widehat{I}_3)}
		= \mathrm{A}+\mathrm{B}+\mathrm{C}
\end{equation*}
with
\begin{align*}
	\mathrm{A} &= \integral{0}{t}{^{3}t}\cumul{\overleftarrow{T}(\widehat{I}_1 \widehat{I}_2 \widehat{I}_3)}=6\iiintegral{0}{t}{t_{1}}{0}{t_{1}}{t_{2}}{0}{t_{2}}{t_{3}}\cumul{\widehat{I}_1 \widehat{I}_2 \widehat{I}_3}
		&\textnormal{(main term)}\,,\\
	\mathrm{B} &= 3\integral{0}{t}{^{2}t}\cumul{\overleftarrow{T}(\widehat{I}_1[\widehat{Q}_2,\widehat{I}_2])}=3\iintegral{0}{t}{t_1}{0}{t_1}{t_2} \cumul{
\widehat{I}_1[\widehat{Q}_2,\widehat{I}_2]+[\widehat{Q}_1,\widehat{I}_1]\widehat{I}_2}
		&\textnormal{(1st contact term)}\,,\\
	\mathrm{C} &= \integral{0}{t}{t_1}\cumul{[\widehat{Q}_1,[\widehat{Q}_1,\widehat{I}_1]]} &\textnormal{(2nd contact term)}\,.
\end{align*}
Here Wick's rule (\ref{eq_416}) is appropriate, whence each term splits into two.\\

\noindent{\bf (a) Main term.} (i) We get
\begin{align*}
	\cumul{\widehat{I}_1\widehat{I}_2\widehat{I}_3} =&\, \tr(\rho I_1\rho' I_2\rho' I_3)-\tr(\rho I_1\rho' I_3\rho I_2)\nonumber\\
		=&\, \frac{1}{(2\pi)^3}\integ{J\times\mathbb{R}\setminus{J}\times\mathbb{R}\setminus{J}}{^{3}k}
			e^{i(E_1-E_2)t_1}e^{i(E_2-E_3)t_2}e^{i(E_3-E_1)t_3}\matel{1}{I}{2}\matel{2}{I}{3}\matel{3}{I}{1}\nonumber\\
		&\,-\frac{1}{(2\pi)^3}\integ{J\times\mathbb{R}\setminus{J}\times J}{^{3}k}
			e^{i(E_1-E_2)t_1}e^{i(E_2-E_3)t_3}e^{i(E_3-E_1)t_2}\matel{1}{I}{2}\matel{2}{I}{3}\matel{3}{I}{1}\,.
\end{align*}
The exponentials of the second term are obtained from those of the first one by exchanging $E_3-E_1$ by $E_2-E_3$, which leaves the sum $E_2-E_1$ unaffected. This symbolic {\it transformation rule} may be deduced from Wick's rule (\ref{eq_416}). For possibly distinct observables and for the terms as a whole it reads:
\begin{itemize}
	\item[-] exchange $\matel{3}{\cdot}{1}$ by $\matel{2}{\cdot}{3}$;
	\item[-] exchange $E_3-E_1$ by $E_2-E_3$;
	\item[-] in the integration domain, replace $k_3\in \mathbb{R}\setminus{J}$ by $k_3\in J$;
    \item[-] apply an overall minus sign.
\end{itemize}
In the present case the first item leaves the integrand unchanged. We shall denote the transformation by $\mathcal{T}_{(23)}$, as it essentially arises by exchanging positions $2$ and $3$ in the product of operators.

The time integrals are given by Eq. (\ref{eq_38}) with $x=E_1-E_2$ and $y=E_2-E_3$ resp. $y=E_3-E_1$. It yields
\begin{equation*}
	\lim\limits_{t\to+\infty}\frac{1}{t}
		\integral{0}{t}{^{3}t}\cumul{\overleftarrow{T}(\widehat{I}_1 \widehat{I}_2 \widehat{I}_3)}
		= \mathrm{A_I}+\mathrm{A_{II}}
\end{equation*}
with
\begin{equation}\label{eq_515}
\begin{aligned}
	\mathrm{A_I} &= -\frac{6}{(2\pi)^3}\integ{J\times\mathbb{R}\setminus{J}\times\mathbb{R}\setminus{J}}{^{3}k}
		\frac{\matel{1}{I}{2}\matel{2}{I}{3}\matel{3}{I}{1}}{(E_1-E_2+i\,0)(E_1-E_3+i\,0)}\,,\\
	\mathrm{A_{II}} &= \mathcal{T}_{(23)}[\mathrm{A_I}]
		= \frac{6}{(2\pi)^3}\integ{J\times\mathbb{R}\setminus{J}\times J}{^{3}k}
		\frac{\matel{1}{I}{2}\matel{2}{I}{3}\matel{3}{I}{1}}{(E_1-E_2+i\,0)(E_3-E_2+i\,0)}\,.
\end{aligned}
\end{equation}

\noindent{\bf(b) Contact terms.} (i) The long--time behavior of the 1st contact term is given by Eq.~(\ref{eq_36}), while the integrand of the 2nd is time--independent. Expanding the commutators we obtain
\begin{align*}
	\lim\limits_{t\to+\infty}\frac{3}{t}\,
		\integral{0}{t}{^{2}t}\cumul{\overleftarrow{T}(\widehat{I}_1[\widehat{Q}_2,\widehat{I}_2])}
		&= \mathrm{B_I}+\mathrm{B_{II}}\,,\\
	\lim\limits_{t\to+\infty}\frac{1}{t}
		\integral{0}{t}{t_1}\cumul{[\widehat{Q}_1,[\widehat{Q}_1,\widehat{I}_1]]}
		&= \mathrm{C_I}+\mathrm{C_{II}}\,,
\end{align*}
with
\begin{multline*}
	\mathrm{B_I} = \frac{3}{(2\pi)^3} \integ{J\times\mathbb{R}\setminus{J}\times\mathbb{R}\setminus{J}}{^{3}k} i
		\left(\frac{\matel{1}{I}{2}\matel{2}{Q}{3}\matel{3}{I}{1}-\matel{1}{I}{2}\matel{2}{I}{3}\matel{3}{Q}{1}}{E_1-E_2+i\,0}
		\right.\\
                \left.
		+\frac{\matel{1}{Q}{2}\matel{2}{I}{3}\matel{3}{I}{1}-\matel{1}{I}{2}\matel{2}{Q}{3}\matel{3}{I}{1}}{E_1-E_3+i\,0}	\right)\,,
\end{multline*}
\begin{equation*}
	\mathrm{C_I} = \frac{1}{(2\pi)^3} \integ{J\times\mathbb{R}\setminus{J}\times\mathbb{R}\setminus{J}}{^{3}k}
		(\matel{1}{Q}{2}\matel{2}{Q}{3}\matel{3}{I}{1}-2\matel{1}{Q}{2}\matel{2}{I}{3}\matel{3}{Q}{1}
		+\matel{1}{I}{2}\matel{2}{Q}{3}\matel{3}{Q}{1})\,.
\end{equation*}
and $\mathrm{B_{II}}=\mathcal{T}_{(23)}[\mathrm{B_I}]$, $\mathrm{C_{II}}=\mathcal{T}_{(23)}[\mathrm{C_I}]$.\\

(ii) We distinguish between matrix elements $\matel{i}{Q}{j}$ as to whether $k_i$ and $k_j$ belong to the same or to different sets among $J$ and $\mathbb{R}\setminus{J}$. In the first instance Eq.~(\ref{eq_507}) applies, whereas in the second its generalization (\ref{eq_427}) is required. For example the first two terms in the integrand of $\mathrm{B_I}$ become
\begin{align*}
	\sum\limits_{s=\pm}\left(\frac{\matel{1}{I}{2}\matel{2}{I}{3}^{(s)}\matel{3}{I}{1}}{(E_1-E_2+i\,0)(E_2-E_3+s\,i\,0)}
	-\frac{\matel{1}{I}{2}\matel{2}{I}{3}^{(s)}\matel{3}{I}{1}}{(E_1-E_2+i\,0)(E_3-E_1-i\,0)}\right)\,,
\end{align*}
where in the second term we used the identity $\matel{2}{I}{3} = \matel{2}{I}{3}^{(+)}+\matel{2}{I}{3}^{(-)}$. In view of the integration domains the splitting will more generally affects $\matel{2}{I}{3}$ in $\mathrm{\alpha_I}$ and $\matel{3}{I}{1}$ in $\mathrm{\alpha_{II}}$, ($\mathrm{\alpha}=\mathrm{A,B,C}$). As a result each term $\mathrm{\alpha_{I,II}}$ may be written as
\begin{align*}
	\mathrm{\alpha_I} &= \frac{1}{(2\pi)^3} \integ{J\times\mathbb{R}\setminus{J}\times\mathbb{R}\setminus{J}}{^{3}k}
		\sum\limits_{s=\pm}\mathrm{\alpha_I}^{(s)}(E_3-E_1, E_2-E_3)\matel{1}{I}{2}\matel{2}{I}{3}^{(s)}\matel{3}{I}{1}\,,\\
	\mathrm{\alpha_{II}} &= \frac{1}{(2\pi)^3} \integ{J\times\mathbb{R}\setminus{J}\times J}{^{3}k}
		\sum\limits_{s=\pm}\mathrm{\alpha_{II}}^{(s)}(E_3-E_1, E_2-E_3)\matel{1}{I}{2}\matel{2}{I}{3}\matel{3}{I}{1}^{(s)}
\end{align*}
for some distributions $\mathrm{\alpha_{I,II}}^{(\pm)}$. We remark that $\mathrm{\alpha_{II}}^{(\pm)}(E_3-E_1, E_2-E_3)=-\mathrm{\alpha_I}^{(\pm)}(E_2-E_3, E_3-E_1)$, because the matrix elements carrying the superscript $(s)$ in the two cases are also those exchanged by the transformation rule $\mathcal{T}_{(23)}$. Moreover, the dependence of $\mathrm{\alpha_{I}}^{(s)}$ on $s$ is of the form
\begin{equation}\label{form1}
	\mathrm{\alpha_I}^{(\pm)}(E_3-E_1,E_2-E_3)
		= \frac{\widehat{\mathrm{\alpha}}_\mathrm{I}(E_3-E_1,E_2-E_3)}{E_1-E_3+i\,0}
		+ \frac{\widecheck{\mathrm{\alpha}}_\mathrm{I}(E_3-E_1,E_2-E_3)}{E_2-E_3\pm i\,0}\,,
\end{equation}
where  $\widehat{\mathrm{\alpha}}_\mathrm{I}$ and $\widecheck{\mathrm{\alpha}}_\mathrm{I}$ ($\mathrm{\alpha =A,B,C}$) are as follows:
\begin{align}
	\widehat{\mathrm{A}}_\mathrm{I}&=-\frac{6}{E_1-E_2+i\,0}\,, &\widecheck{\mathrm{A}}_\mathrm{I} &= 0\,,\nonumber\\
	\widehat{\mathrm{B}}_\mathrm{I}&=\frac{3}{E_1-E_2+i\,0}+\frac{3}{E_1-E_2-i\,0}\,, &\quad
		\widecheck{\mathrm{B}}_\mathrm{I} &= \frac{3}{E_1-E_2+i\,0}-\frac{3}{E_1-E_3+i\,0}\,,\label{tab}\\
	\widehat{\mathrm{C}}_\mathrm{I}&=-\frac{2}{E_1-E_2-i\,0}\,,&
		\widecheck{\mathrm{C}}_\mathrm{I} &= -\frac{1}{E_1-E_2-i\,0}+\frac{1}{E_1-E_3+i\,0}\,.\nonumber
\end{align}
The claim we are heading to is
\begin{equation}
\label{eq_519}
	\mathrm{A_I}+\mathrm{B_I}+\mathrm{C_I}=\frac{1}{2 \pi} \integral{\mur}{\mul}{E} T(E)R(E)^2\,,
	\qquad \mathrm{A_{II}}+\mathrm{B_{II}}+\mathrm{C_{II}}=-\frac{1}{2 \pi} \integral{\mur}{\mul}{E} T(E)^2 R(E)\,,
\end{equation}
which leads to binomial statistics (\ref{eq_414}) in view of $TR^2-T^2 R = T(1-T)(1-2T)$. To establish it, we observe that the sum,
\begin{equation}
\label{eq_520}
	\mathrm{A_I}+\mathrm{B_I}+\mathrm{C_I}
		= \frac{1}{(2\pi)^3} \integ{J\times\mathbb{R}\setminus{J}\times\mathbb{R}\setminus{J}}{^{3}k}
		\sum\limits_{s=\pm}\mathrm{\Delta_I}^{(s)}\matel{1}{I}{2}\matel{2}{I}{3}^{(s)}\matel{3}{I}{1}\,,
\end{equation}
likewise involves distributions $\mathrm{\Delta_I}^{(\pm)}$ of the form (\ref{form1}) with
\begin{equation*}
\widehat{\mathrm{\Delta}}_\mathrm{I}=2\pi i\delta(E_1-E_2)-\frac{2}{E_1-E_2+i\,0}\,,\qquad
\widecheck{\mathrm{\Delta}}_\mathrm{I}=-2\pi i\delta(E_1-E_2)+\frac{2(E_2-E_3)}{(E_1-E_2+i\,0)(E_1-E_3+i\,0)}\,.
\end{equation*}
This is seen by summing terms within the columns of table (\ref{tab}) and by using the SW formula (\ref{ws}). The second term of $\widecheck{\mathrm{\Delta}}_\mathrm{I}$ is a distribution with poles at $E_2-i\,0$, $E_3-i\,0$ not pinching the $E_1$--axis. It thus vanishes to first order at $E_2=E_3$ and cancels in $\mathrm{\Delta_I}^{(\pm)}$ against the second term of $\widehat{\mathrm{\Delta}}_\mathrm{I}$. We are thus left with
\begin{align}\label{eq_521}
			\mathrm{\Delta_I}^{(\pm)}
= 2\pi i\,\delta(E_1-E_2)\left(\frac{1}{E_2-E_3+i\,0}-\frac{1}{E_2-E_3\pm i\,0}\right)\,,
\end{align}
and we conclude that $\mathrm{\Delta_I}^{(+)}=0$ and $\mathrm{\Delta_I}^{(-)}=(2\pi)^2\delta(E_1-E_2)\delta(E_2-E_3)$. The conditions $E_1=E_2$ and $E_2=E_3$ are satisfied along the diagonals $k_1=\pm k_2$ resp.~$k_2=\pm k_3$. That happens jointly and within the integration domain only for $k_1=-k_2=-k_3$ with $k_1\in [\kr,\kl]$, whence
\begin{equation*}
	\mathrm{A_I}+\mathrm{B_I}+\mathrm{C_I}
		= \frac{1}{2\pi}\integral{\kr}{\kl}{k_1}\frac{\vert\matel{1}{I}{-1}\vert^2\matel{-1}{I}{-1}^{(-)}}{(2k_1)^2}
		= \frac{1}{2 \pi} \integral{\mur}{\mul}{E} T(E)R(E)^2\,,
\end{equation*}
as claimed. The last equality is by Eqs.~(\ref{eq_423}, \ref{eq_431}) with $E=k_1^2$.

Similarly, by $\mathrm{\Delta_{II}}^{(\pm)}(E_3-E_1,E_2-E_3)=-\mathrm{\Delta_I^{(\pm)}}(E_2-E_3,E_3-E_1)$, we obtain $\mathrm{\Delta_{II}}^{(+)}=0$ and $\mathrm{\Delta_{II}}^{(-)}=-(2\pi)^2\delta(E_1-E_2)\delta(E_1-E_3)$. Hence
\begin{equation*}
	\mathrm{A_{II}}+\mathrm{B_{II}}+\mathrm{C_{II}}
		= -\frac{1}{2\pi}\integral{\kr}{\kl}{k_1}\frac{\vert\matel{1}{I}{-1}\vert^2\matel{1}{I}{1}^{(-)}}{(2k_1)^2}
		= -\frac{1}{2 \pi} \integral{\mur}{\mul}{E} T(E)^2 R(E)\,,
\end{equation*}
where the last equality follows by Eqs.~(\ref{eq_423}, \ref{eq_424}) with $E=k_1^2$.\\

\noindent{\bf 3rd cumulant of the second kind.} The cumulant is given in Eq.~(\ref{eq_412}) as
\begin{equation}
\label{eq_524}
	\cumul{(\Delta Q)^3} = \frac{1}{8}\integral{0}{t}{^{3}t}\bigl(
		\cumul{\overrightarrow{T}^{*}(\widehat{I}_1 \widehat{I}_2 \widehat{I}_3)}
		+ 3\, \cumul{\overrightarrow{T}^{*}(\widehat{I}_1 \widehat{I}_2)\widehat{I}_3}
		+ 3\, \cumul{\widehat{I}_1\overleftarrow{T}^{*}(\widehat{I}_2 \widehat{I}_3)}
		+ \cumul{\overleftarrow{T}^{*}(\widehat{I}_1 \widehat{I}_2 \widehat{I}_3)}\bigr)\,.
\end{equation}
A preliminary observation is useful. By $\overline{\langle\widehat{A}\rangle}=\langle\widehat{A}^*\rangle$ we have $\overline{\langle\widehat{A}\widehat{B}\rangle}=\langle\widehat{B}^*\widehat{A}^*\rangle$, $\overline{\langle\overleftarrow{T}(\widehat{A}(t_1)\widehat{B}(t_2))\rangle}=\langle\overrightarrow{T}(\widehat{A}(t_1)^*\widehat{B}(t_2)^*)\rangle$, and likewise for higher products, $T^*$--ordered products, and cumulants. Given that in the above expression the currents are self-adjoint and the $t_i$'s dummy variables, the two extreme and the two middle terms are so related. The missing item for establishing binomial statistics here is thus just
\begin{equation}
\label{eq_525}
	\lim\limits_{t\to +\infty}\frac{1}{t}\integral{0}{t}{^{3}t}
	\cumul{\widehat{I}_1\overleftarrow{T}^{*}(\widehat{I}_2 \widehat{I}_3)}
	=\frac{1}{2 \pi} \integral{\mur}{\mul}{E} T(E)\left( 1-T(E)\right)\left( 1-2T(E)\right)\,.
\end{equation}
The computation is similar to that of the 3rd cumulant of the first kind, whence we refer to it for more details. By Eq.~(\ref{eq_405}) we have
\begin{equation*}
	\integral{0}{t}{^{3}t}\cumul{\widehat{I}_1\overleftarrow{T}^{*}(\widehat{I}_2 \widehat{I}_3)} = \mathrm{D}+\mathrm{E}
\end{equation*}
with
\begin{equation}\label{eq_527}
\begin{aligned}
	\mathrm{D} &= \integral{0}{t}{^{3}t}\cumul{\widehat{I}_1\overleftarrow{T}(\widehat{I}_2 \widehat{I}_3)}
		&\textnormal{(main term)}\,,\\
	\mathrm{E} &= \integral{0}{t}{^{2}t}\cumul{\widehat{I}_1 [\widehat{Q}_2,\widehat{I}_2])}
		&\textnormal{(contact term)}\,.
\end{aligned}
\end{equation}
We apply Wick's rule (\ref{eq_416}) to both terms. The long--time behavior of the main term is extracted by Eq.~(\ref{eq_44}); that of the contact term by Eq.~(\ref{eq_37}). Expanding the commutators we so obtain
\begin{align*}
	\lim\limits_{t\to+\infty}\frac{1}{t}
		\integral{0}{t}{^{3}t}\cumul{\widehat{I}_1\overleftarrow{T}(\widehat{I}_2 \widehat{I}_3)}
		&= \mathrm{D_I}+\mathrm{D_{II}}\,,\\
	\lim\limits_{t\to+\infty}\frac{1}{t}
		\integral{0}{t}{^{2}t}\cumul{\widehat{I}_1 [\widehat{Q}_2,\widehat{I}_2])} &=\mathrm{E_I}+\mathrm{E_{II}}\,,
\end{align*}
with
\begin{align}
\label{eq_531}
	\mathrm{D_I} &= \frac{2}{(2\pi)^3}\integ{J\times\mathbb{R}\setminus{J}\times\mathbb{R}\setminus{J}}{^{3}k}
		2\pi\,i\,\delta(E_1-E_2)\frac{\matel{1}{I}{2}\matel{2}{I}{3}\matel{3}{I}{1}}{E_2-E_3+i\,0}\,,\\
\nonumber
	\mathrm{E_I} &= \frac{1}{(2\pi)^3}\integ{J\times\mathbb{R}\setminus{J}\times\mathbb{R}\setminus{J}}{^{3}k}
		2\pi\delta(E_1-E_2)(\matel{1}{I}{2}\matel{2}{Q}{3}\matel{3}{I}{1}-\matel{1}{I}{2}\matel{2}{I}{3}\matel{3}{Q}{1})\,.
\end{align}
$\mathrm{D_{II}}$ is obtained from $\mathrm{D_I}$ by the rule $\mathcal{T}_{(23)}$ introduced in relation with the previous cumulant. Likewise for $\mathrm{E_{II}}$ and $\mathrm{E_I}$. In analogy with the claim (\ref{eq_519}) made there, the present one is
\begin{equation}
\label{eq_532}
	\mathrm{D_I}+\mathrm{E_I}=\frac{1}{2 \pi} \integral{\mur}{\mul}{E} T(E)R(E)^2\,,
	\qquad \mathrm{D_{II}}+\mathrm{E_{II}}=-\frac{1}{2 \pi} \integral{\mur}{\mul}{E} T(E)^2 R(E)\,.
\end{equation}
By the same steps as those leading to Eq.~(\ref{eq_520}), we find
\begin{equation*}
	\mathrm{D_I}+\mathrm{E_I} = \frac{1}{(2\pi)^3} \integ{J\times\mathbb{R}\setminus{J}\times\mathbb{R}\setminus{J}}{^{3}k}
		\sum\limits_{s=\pm}\mathrm{\Gamma_I}^{(s)}(k_1,k_2,k_3)\matel{1}{I}{2}\matel{2}{I}{3}^{(s)}\matel{3}{I}{1}\,,
\end{equation*}
with the distributions
\begin{equation*}
	\mathrm{\Gamma_I}^{(\pm)} = 2\pi i\,\delta(E_1-E_2)\left(\frac{1}{E_2-E_3+i\,0}-\frac{1}{E_2-E_3\pm i\,0}\right)
		=\mathrm{\Delta_I}^{(\pm)}\,.
\end{equation*}
Hence this case reduces to the 3rd cumulant of the first kind (\ref{eq_521}), which establishes the claims~(\ref{eq_532}).

\subsection{The limit of a large detector}
\label{subsec_smeared_proj}

We will consider the situation of a detector extending over a region much larger than the Fermi wavelength. Clearly, the binomial distribution persists, this situation being a limiting case of the one dealt with before. The point though to be made is (i) that the contact terms in Eqs.~(\ref{eq_409}-\ref{eq_412}) vanish in the limit. Put differently, Matthews' time--ordered correlators reduce to ordinary time--ordered correlators, which alone account for the binomial distribution. Moreover, (ii) we provide an independent derivation of that latter fact. We shall analyze the cumulants separately, though in very similar manners. The values of some integrals used along the way are collected at the end of the section.

The large detector is modeled by means of scaling. Let $Q_0(x)$ be a fixed function satisfying (\ref{eq_25}). We choose the profile of the detector to be given by the function
\begin{equation*}
Q(x)= Q_{0}(x/l)\,,
\end{equation*}
which for $l\ge 1$ retains that property, and consider it in the limit $l\to\infty$. The scaling implies
\begin{align}
\label{eq_141}
	Q'(x) = l^{-1} Q_{0}'(x/l)\,,\qquad
	\widehat{Q'}(k) = \widehat{Q_{0}'}(lk)\,,\qquad
	\widehat{Q'^{2}}(k)=l^{-1}\widehat{(Q_{0}')^{2}}(lk)\, .
\end{align}
\noindent{\bf 2nd cumulant of the first kind.}
(i) We first show that the contact term vanishes in the limit $l\to\infty$:
\begin{equation*}
	\lim\limits_{t\to\infty}\frac{1}{t}\integral{0}{t}{t_1} \cumul{[\widehat{Q}_1,\widehat{I}_1]}\;\xrightarrow[l\to+\infty]{}\;0\,.
\end{equation*}
The limit $t\to\infty$ is superfluous, since $\cumul{[\widehat{Q}_1,\widehat{I}_1]}$ is independent of $t_1$ and in fact by (\ref{eq_925}) equal to
\begin{equation*}
\cumul{\widehat{[Q,I]}}=2i\tr(\rho Q'(x)^{2})=\frac{2i}{2\pi}\integral{-\kr}{\kl}{k_1}\matel{1}{Q'^2}{1}\,.
\end{equation*}
Using (\ref{eq_B2}, \ref{eq_B3}) for $Q'^2$ instead of $Q$, the matrix element is seen to be a linear combination of $\widehat{Q'^2}(k)$ for $k=0,\pm 2k_1$. They are of order $O(l^{-1})$ by (\ref{eq_141}), proving the first claim.\\

(ii) Let us now come to the main term:
\begin{equation}
\label{eq_149}
	\lim\limits_{t\to\infty}\frac{1}{t}\integral{0}{t}{^{2}t}\cumul{\overleftarrow{T}(\widehat{I}_1 \widehat{I}_2)} = \frac{2i}{(2 \pi)^{2}}
		\integ{J\times\mathbb{R}\setminus{J}}{^{2}k} \frac{\vert\matel{1}{I}{2}\vert^{2}}{E_{1}-E_{2}+i\,0}	\;\xrightarrow[l\to+\infty]{}\;
\frac{1}{2 \pi} \integral{\mur}{\mul}{E} T(E)\,(1-T(E))
\,.
\end{equation}
The equality was shown in (\ref{eq_504}, \ref{eq_506}), whereas the limit is the second claim being made here.

The regularization $+i\,0$ of the denominator only matters when $E_1=E_2$, i.e. on the diagonals $k_1=\pm k_2$, and, once restricted to the integration domain, only for $k_1=-k_2\in [\kr,\kl]$ (see Fig.~\ref{fig_2}). Moreover, the matrix element $\matel{1}{I}{2}$ is a linear combination of $\widehat{Q'}(\pm k_1 \pm k_2)$. In the limit $l\to\infty$ their supports concentrate by Eq.~(\ref{eq_141}) near the same diagonals, which by the same token get restricted to part of just one. This allows to:
\begin{itemize}
\item[-] Use the factorization $E_{1}-E_{2}+i\,0 = (k_{1}-k_{2})(k_{1}+k_{2}+i\,0)$, as appropriate for $k_1>0$, $k_2<0$.
\item[-] Select the corresponding expression (\ref{me+-}) for $\matel{i}{I}{j}$ from Appendix~\ref{sec_matrix_elements}; and therein neglect any terms vanishing on that diagonal. Hence, effectively,
\begin{align}
\label{eq_149a}
\matel{1}{I}{2}=(k_1-k_2)\overline{t(k_1)}r(k_2)\widehat{Q'}(k_1+k_2)\,.
\end{align}
\end{itemize}
The integrand of Eq.~(\ref{eq_149}) so becomes
\begin{equation*}
		(k_{1}-k_{2})T(E_{1}) R(E_{2})\frac{|\widehat{Q'}(k_1+k_2)|^2}{k_{1}+k_{2}+i\,0}\,.
\end{equation*}
The last factor depends on $l$ in the way seen in Eq.~(\ref{eq_41}) for $\rho(x)=|\widehat{Q_{0}'}(x)|^2=\widehat{Q_{0}'}(-x)\widehat{Q_{0}'}(x)$. It can thus be replaced in the limit by $C_{-} \delta(k_{1}+k_{2})$, where
\begin{equation}
\label{eq_620}
C_{-} =\integral{}{}{u} \frac{\widehat{Q_{0}'}(-v)\widehat{Q_{0}'}(v)}{v+i\,0}\,.
 \end{equation}
Accepting for now that $C_{-}=-i\pi$, we obtain the limit (\ref{eq_149}) by means of $2k_1\mathrm{d}k_1=\mathrm{d}E_1$ and by the earlier remark restricting  $k_1$ to $[\kr,\kl]$.\\

\noindent{\bf 3rd cumulant of the first kind.}
(i) We first show that the contact terms vanish in the limit $l\to\infty$:
\begin{equation*}
	\lim\limits_{t\to\infty}\frac{1}{t}\integral{0}{t}{^{2}t}
		\cumul{\overleftarrow{T}(\widehat{I}_1[\widehat{Q}_2,\widehat{I}_2])}\;\xrightarrow[l\to+\infty]{}\;0\,,
		\qquad \lim\limits_{t\to\infty}\frac{1}{t}\integral{0}{t}{t_1}\cumul{[\widehat{Q}_1,[\widehat{Q}_1,\widehat{I}_1]]}
		\;\xrightarrow[l\to+\infty]{}\;0\,.
\end{equation*}
The limit $t\to\infty$ is superfluous in the second claim, since the integrand is time--independent. It actually vanishes even at finite $l$ because
\begin{equation*}
	\cumul{[\widehat{Q},[\widehat{Q},\widehat{I}]]} = \cumul{[\widehat{Q},\widehat{[Q,I]}]} = 2i\,\tr(\rho[Q,Q'^2])=0\,.
\end{equation*}
The second equality is by Eq.~(\ref{eq_925}) and the last one by the vanishing commutator. Turning to the 1st contact term, it is convenient to use the identity $\cumul{\widehat{I}_1[\widehat{Q}_2,\widehat{I}_2]} =\cumul{\widehat{I}_1\widehat{[Q_2,I_2]}}$. By Wick's rule (\ref{eq_415}) and Eq.~(\ref{eq_925}) it may then be recast as
\begin{equation*}
	\lim\limits_{t\to\infty}\frac{1}{t}\integral{0}{t}{^{2}t}
		\cumul{\overleftarrow{T}(\widehat{I}_1\widehat{[Q_2,I_2]})}
		= -\frac{2}{(2\pi)^2}\integ{J\times\mathbb{R}\setminus{J}}{^{2}k}
		\frac{\matel{1}{I}{2}\matel{2}{Q'^2}{1}+\matel{1}{Q'^2}{2}\matel{2}{I}{1}}{E_1-E_2+i\,0}\,.
\end{equation*}
By the results of Appendix \ref{sec_matrix_elements}, the numerator is a linear combination of $\widehat{Q'^2}(\pm k_1\pm k_2)\widehat{Q'}(\pm k_1\pm k_2)$ with various sign combinations. They are of order $O(l^{-1})$ by (\ref{eq_141}), proving the second claim.\\

(ii) Let us now come to the main term. We showed in Eqs.~(\ref{eq_515}) that
\begin{equation*}
\lim\limits_{t\to\infty}\frac{1}{t}\integral{0}{t}{^{3}t}\cumul{\overleftarrow{T}(\widehat{I}_1 \widehat{I}_2 \widehat{I}_3)}
=\mathrm{A_I}+\mathrm{A_{II}}\,,
\end{equation*}
with
\begin{align*}
\mathrm{A_I}&=-\frac{6}{(2\pi)^{3}}\integ{J\times\mathbb{R}\setminus J\times\mathbb{R}\setminus J}{^{3}k}
		\frac{\matel{1}{I}{2}\matel{2}{I}{3}\matel{3}{I}{1}}{(E_1-E_2+i\,0)(E_1-E_3+i\,0)}\,,\\
\mathrm{A_{II}}&=\frac{6}{(2\pi)^{3}}\integ{J\times\mathbb{R}\setminus J\times J}{^{3}k}
\frac{\matel{1}{I}{2}\matel{2}{I}{3}\matel{3}{I}{1}}{(E_1-E_2+i\,0)(E_3-E_2+i\,0)}\,.	
\end{align*}
The claim is now
\begin{equation}
\label{eq_156a}
\mathrm{A_I}	\;\xrightarrow[l\to+\infty]{}\;
\frac{1}{2 \pi} \integral{\mur}{\mul}{E} T(E)R(E)^2\,,\qquad \mathrm{A_{II}}	\;\xrightarrow[l\to+\infty]{}\;
-\frac{1}{2 \pi} \integral{\mur}{\mul}{E}T(E)^2R(E)\,.
\end{equation}
It independently confirms binomial statistics, in view of $TR^2-T^2R=T(1-T)(1-2T)$.

The computation is similar to that of the 2nd cumulant, whence we refer to the discussion following (\ref{eq_149}) for more details. We first discuss term $\mathrm{A_I}$. The regularization of the denominator only matters when $E_1=E_2$ or $E_1=E_3$ and, once the integration domain is taken into account, only for $k_1 = -k_2$ or $k_1 = -k_3$, both along $k_1\in [\kr,\kl]$. A matrix element $\matel{i}{I}{j}$, $(i\neq j)$ concentrates near the planes $k_i=\pm k_j$ as $l\to+\infty$; and their product near the intersection: $k_1 = -k_2 = -k_3$ with $k_1\in [\kr,\kl]$. This allows to:
\begin{itemize}
\item[-] Use the factorizations $E_1-E_j+i\,0=(k_1-k_j)(k_1+k_j+i\,0)$, $(j=2,3)$, as appropriate for $k_1>0$, $k_2,k_3<0$.
\item[-] Select the appropriate expressions for $\matel{i}{I}{j}$ and simplify them as done for Eq.~(\ref{eq_149a}). Hence, we also have from (\ref{me--}, \ref{me-+})
\begin{align}
\label{eq_616}
	\matel{2}{I}{3}&=(k_2+k_3)\bigl( \widehat{Q'}(k_2-k_3)-\overline{r(k_2)}r(k_3)\widehat{Q'}(k_3-k_2)\bigr)\,,\\
\label{eq_617}
	\matel{3}{I}{1}&=(k_1-k_3)t(k_1)\overline{r(k_3)}\widehat{Q'}(-k_3-k_1)\;.
\end{align}
\end{itemize}
The integrand of $\mathrm{A_I}$ is thus recast as
\begin{equation*}
	(k_2+k_3)\,T(E_1)\frac{\widehat{Q'}(k_1+k_2)\,\widehat{Q'}(-k_1-k_3)}{(k_1+k_2+i\,0)(k_1+k_3+i\,0)}\,
		\bigl( r(k_2)\,\overline{r(k_3)}\,\widehat{Q'}(k_2-k_3)-R(E_2)\,R(E_3)\,\widehat{Q'}(k_3-k_2)\bigr)\,.
\end{equation*}
As mentioned, the expression depends on $l$ through $\widehat{Q'}(k) = \widehat{Q_{0}'}(lk)$; moreover, it consists of two terms, to each of which Eq.~(\ref{eq_42}) is applicable by the following observation. Each term contains a product of distributions in the variables $x=k_1+k_2$, $y=k_1+k_3$, and $x-y=k_2-k_3$. The distributions correspond to
\begin{equation*}
\rho_1(x)=\widehat{Q_{0}'}(x)\,\qquad\rho_2(x)=\widehat{Q_{0}'}(-x)\,\qquad\rho_3(x)=\widehat{Q_{0}'}(\pm x)\,,
\end{equation*}
where the $\pm$ refers to the first, resp. second term. In the limit $l\to+\infty$, the integrand thus reduces to
\begin{equation}
\label{eq_604}
	(k_2+k_3)\,T(E_1)\bigl( \widetilde{C}_+r(k_2)\,\overline{r(k_3)}- \widetilde{C}_-R(E_2)\,R(E_3)\bigr)
\delta(k_1+k_2)\delta(k_1+k_3)\,,
\end{equation}
where
\begin{equation*}
\widetilde{C}_\pm=\int \mathrm{d}u \mathrm{d}v
\frac{\widehat{Q_{0}'}(u)\widehat{Q_{0}'}(-v)\widehat{Q_{0}'}(\pm(u-v))}{(u+i\,0)(v+i\,0)}\,.
\end{equation*}
Accepting for now that $\widetilde{C}_+=0$ and $\widetilde{C}_{-}=-(2\pi)^2/6$, we obtain the first limit (\ref{eq_156a}) by means of $2k_1\mathrm{d}k_1=\mathrm{d}E_1$ and by the earlier remark restricting $k_1$ to $[\kr,\kl]$.\\

We now turn to $\mathrm{A_{II}}$. In view of the integration domain the integrand of $\mathrm{A_{II}}$ is supported in the limit $l\to+\infty$ near the segment $k_1=-k_2=k_3$ with $k_1\in[\kr,\kl]$.
We may thus:
\begin{itemize}
	\item[-] Use the factorizations $E_j-E_2+i\,0=(k_j-k_2)(k_j+k_2+i\,0)$ ($j=1,3$) as appropriate for $k_1,k_3>0$ and $k_2<0$.
	\item[-] Select and simplify the relevant matrix elements of current $\matel{i}{I}{j}$ as was done for Eq.~(\ref{eq_149a}). Hence we also have from (\ref{me-+}, \ref{eq_B4}):
	\begin{align}
	\label{eq_624}
		\matel{2}{I}{3} &= (k_3-k_2)t(k_3)\overline{r(k_2)}\widehat{Q'}(-k_2-k_3)\\
	\label{eq_625}
		\matel{3}{I}{1} &= (k_3+k_1)\overline{t(k_3)}t(k_1)\widehat{Q'}(k_3-k_1)\,.
	\end{align}
\end{itemize}
The integrand of $\mathrm{A_{II}}$ then reduces to
\begin{equation*}
	(k_1+k_3)T(E_1)T(E_3)R(E_2)
	\frac{\widehat{Q'}(k_1+k_2)\widehat{Q'}(-k_2-k_3)\widehat{Q'}(k_3-k_1)}{(k_1+k_2+i\,0)(k_2+k_3+i\,0)}\,.
\end{equation*}
In view of $\widehat{Q'}(k)=\widehat{Q_0'}(lk)$ Eq.~(\ref{eq_42}) may be applied with $x=k_1+k_2$, $y=k_2+k_3$ and $x-y=k_1-k_3$. Comparing with the derivation of the integrand (\ref{eq_604}) of $\mathrm{A_I}$, that of $\mathrm{A_{II}}$ reduces in the limit to
\begin{equation*}
	\widetilde{C}_{-} (k_1+k_3)T(E_1)T(E_3)R(E_2)\delta(k_1+k_2)\delta(k_2+k_3)\,.
\end{equation*}
Now the second claim (\ref{eq_156a}) follows just as the first one did from (\ref{eq_604}).\\

\noindent{\bf 3rd cumulant of the second kind.} By the earlier observation following (\ref{eq_524}) we only need to investigate
\begin{equation*}
	\integral{0}{t}{^{3}t}\cumul{\widehat{I}_1\overleftarrow{T}^{*}(\widehat{I}_2 \widehat{I}_3)}
		= \integral{0}{t}{^{3}t}\cumul{\widehat{I}_1\overleftarrow{T}(\widehat{I}_2 \widehat{I}_3)}
		+\integral{0}{t}{^{2}t}\cumul{\widehat{I}_1[\widehat{Q}_2,\widehat{I}_2]}\,.
\end{equation*}

(i) We first show that the contact term vanishes as $l\to\infty$. Using $\cumul{\widehat{I}_1[\widehat{Q}_2,\widehat{I}_2]} =\cumul{\widehat{I}_1\widehat{[Q_2,I_2]}}$, the time integral is given by Eq.~(\ref{eq_37}). Hence, by Eq.~(\ref{eq_925}), we have
\begin{equation*}
	\lim\limits_{t\to+\infty}\frac{1}{t}\integral{0}{t}{^{2}t}\cumul{\widehat{I}_1 \widehat{[Q_2,I_2]})}
		= \frac{2i}{2\pi}\integ{J\times\mathbb{R}\setminus{J}}{^{2}k}\delta(E_1-E_2)\matel{1}{I}{2}\matel{2}{Q'^2}{1}\,.
\end{equation*}
According to Eq.~(\ref{me+-}) and Eq.~(\ref{meq+-}) with $Q'^2$ substituted for $Q$, the integrand is a linear combination of $\widehat{Q'^2}(\pm k_1\pm k_2)\widehat{Q'}(\pm k_1\pm k_2)$ with various sign combinations. By (\ref{eq_141}) they are of order $O(l^{-1})$, proving the claim.\\

(ii) We now come to the main term. We showed
\begin{equation*}
	\lim\limits_{t\to +\infty}\frac{1}{t}\integral{0}{t}{^{3}t}
	\cumul{\widehat{I}_1\overleftarrow{T}(\widehat{I}_2 \widehat{I}_3)}
	= \mathrm{D_I}+\mathrm{D_{II}}\,,
\end{equation*}
with
\begin{align*}
	\mathrm{D_I} &= \frac{2i}{(2\pi)^2}\int_{\kr}^{\kl}\frac{\mathrm{d}k_1}{2k_1}\integ{\mathbb{R}\setminus{J}}{k_3}
		\frac{\matel{1}{I}{-1}\matel{-1}{I}{3}\matel{3}{I}{1}}{E_1-E_3+i\,0}\,,\\
	\mathrm{D_{II}} &= -\frac{2i}{(2\pi)^2}\int_{\kr}^{\kl}\frac{\mathrm{d}k_1}{2k_1}\integ{J}{k_3}
		\frac{\matel{1}{I}{-1}\matel{-1}{I}{3}\matel{3}{I}{1}}{E_3-E_1+i\,0}\,,
\end{align*}
where we evaluated the delta distributions in Eq.~(\ref{eq_531}). The claim now is
\begin{equation}
\label{eq_615}
	\mathrm{D_I}\;\xrightarrow[l\to+\infty]{}\;\frac{1}{2 \pi} \integral{\mur}{\mul}{E} T(E)R(E)^2\,,\qquad
	\mathrm{D_{II}}	\;\xrightarrow[l\to+\infty]{}\;-\frac{1}{2 \pi} \integral{\mur}{\mul}{E}T(E)^2R(E)\,,
\end{equation}
which is sufficient in view of $TR^2-T^2R=T(1-T)(1-2T)$. We first consider $\mathrm{D_I}$. Taking the integration domain into account, the support of its integrand concentrates in the limit near the diagonal $k_1=-k_3$ for $k_1\in [\kr,\kl]$. This allows to:
\begin{itemize}
	\item[-] Use the factorization $E_1-E_3+i\,0=(k_1-k_3)(k_1+k_3+i\,0)$ as appropriate for $k_1>0$ and $k_3<0$.
	\item[-] Select the appropriate matrix elements of currents and simplify them as was done in Eq.~(\ref{eq_617}); moreover, substituting $-k_1$ for $k_2$ in Eq.~(\ref{eq_616}) we obtain
	\begin{equation*}
		\matel{-1}{I}{3} = (k_3-k_1)\bigl( \widehat{Q'}(-k_1-k_3)-\overline{r(k_1)}r(k_3)\widehat{Q'}(k_1+k_3)\bigr)\,;
	\end{equation*}
	$\matel{1}{I}{-1}$ is given by Eq.~(\ref{eq_422}).
\end{itemize}
The integrand of $\mathrm{D_I}$ may thus be recast as
\begin{equation*}
	(k_3-k_1)\,T(E_1)\frac{\widehat{Q'}(-k_1-k_3)}{(k_1+k_3+i\,0)}\,
		\bigl( r(k_1)\,\overline{r(k_3)}\,\widehat{Q'}(-k_1-k_3)-R(E_1)\,R(E_3)\,\widehat{Q'}(k_1+k_3)\bigr)\,.
\end{equation*}
As mentioned, the expression depends on $l$ through $\widehat{Q'}(k)=\widehat{Q_{0}'}(lk)$, whence Eq.~(\ref{eq_41}) may be applied with $x=k_1+k_3$. In the limit $l\to+\infty$ the above thus reduces to
\begin{equation*}
	(k_3-k_1)\delta(k_1+k_3)\,
		\bigl( C_+\,r(k_1)\,\overline{r(k_3)}- C_{-}\,R(E_1)\,R(E_3))\,,
\end{equation*}
with the constants $C_{-}$ and $C_+$ given by Eq.~(\ref{eq_620}) resp. by
\begin{equation*}
	C_+ =\integral{}{}{u} \frac{(\widehat{Q_{0}'}(-v))^2}{v+i\,0}\,.
\end{equation*}
Accepting on top of $C_{-}=-i\pi$ that $C_+=0$, the first limit (\ref{eq_615}) follows by $2k_1\mathrm{d}k_1=\mathrm{d}E$.\\

We now turn to $\mathrm{D_{II}}$. Here the support of the integrand will concentrate near $k_1=k_3$ for $k_1\in [\kr,\kl]$ as $l\to+\infty$. The denominator therefore factorizes as $E_3-E_1+i\,0=(k_1+k_3)(k_3-k_1+i\,0)$, as appropriate for $k_1,k_3>0$; the relevant matrix elements are given by Eqs.~(\ref{eq_422}, \ref{eq_625}) and Eq.~(\ref{eq_624}) with $-k_1$ substituted for $k_2$. The integrand of $\mathrm{D_{II}}$ thus becomes
\begin{equation*}
	(k_1+k_3)T(E_1)T(E_3)R(E_1)
	\frac{\widehat{Q'}(k_1-k_3)\,\widehat{Q'}(k_3-k_1)}{(k_3-k_1+i\,0)}\,.
\end{equation*}
By Eq.~(\ref{eq_41}) it reduces in the limit to $C_{-} \delta(k_1-k_3)(k_1+k_3)T(E_1)^2R(E_1)$ , which establishes the second limit (\ref{eq_615}) by the aforementioned identity $C_{-}=-i\pi$.\\

To close this section, we compute the four integrals encountered along the way. Let us generalize $C_\pm$ to $C_\pm=C_\pm(0)$ where
\begin{equation*}
  C_\pm(u)=\integral{}{}{v}\frac{\widehat{Q_{0}'}(-v)\widehat{Q_{0}'}(\pm (u-v))}{v+i\,0}\,,
\end{equation*}
for which we claim
\begin{equation*}
C_+(u)=0\,,\qquad
C_{-}(u)=-2\pi i \widehat{Q_{0}Q_{0}'}(-u)\,.
\end{equation*}
Indeed, by Eq.~(\ref{eq_425}) and Parseval's identity we have
 \begin{equation*}
  C_\pm(u)=(-i)\integral{}{}{v}\widehat{Q_{0}}(-v)\widehat{Q_{0}'}(\pm (u-v))=
-2\pi i \integral{}{}{x}Q_{0}(x)Q_{0}'(\mp x)e^{\mp iux}\,.
\end{equation*}
The first result follows by the support properties of $Q_{0}$; the second is now read off and its special case $C_{-}(0)=-i\pi$ is by
$\integral{}{}{x}Q_{0}(x)Q_{0}'(x)=1/2$. Let us now come to
\begin{equation*}
\widetilde{C}_\pm =\integral{}{}{u}\frac{\widehat{Q_{0}'}(u) C_\pm(u)}{u+i\,0}\,.
\end{equation*}
Clearly $\widetilde{C}_+=0$. By Eq.~(\ref{eq_425}) and the SW formula (\ref{ws}) we have
\begin{align*}
\widetilde{C}_{-} =\integral{}{}{u} \bigl(\frac{1}{u-i\,0}-2\pi i\,\delta(u)\bigr)
		\widehat{Q_{0}'}(u)C_{-}(u)
   =2\pi\integral{}{}{u}\widehat{Q_{0}}(u)\widehat{Q_{0}Q_{0}'}(-u)-\frac{(2\pi)^2}{2}\widehat{Q_{0}'}(0)
=-\frac{(2\pi)^2}{6}\,,
\end{align*}
where the last equality is by Parseval's identity followed by $\integral{}{}{x}Q_{0}^2(x)Q_{0}'(x)=1/3$, as well as by $\widehat{Q_{0}'}(0)=\integral{}{}{x}Q_{0}'(x)=1$.

\subsection{The case of a linear dispersion relation}
\label{subsec_limit_linear}

We shall consider the limiting case of a linear dispersion relation described in Sect.~\ref{subsec_linear_model}. It arises in the limit $\lambda\to 0$ of vanishing Fermi wavelength and thus ought to retain the binomial character of the transport statistics. In the model the scatterer and the detector are pointlike objects on the real line. We treat the cases where they are separated by $l>0$, respectively coincident ($l=0$). In the following we briefly review the results obtained in \cite{graf:09} for the 3rd cumulant of the first kind, before extending them to the 3rd cumulant of the second kind. \\

\noindent{\bf 3rd cumulant of the first kind.} It was shown in~\cite{graf:09} that the 3rd cumulant of the first kind (\ref{eq_410}) exhibits binomial statistics,
\begin{equation*}
	\lim\limits_{t\to +\infty}\frac{1}{t}\integral{0}{t}{^3 t}
	\cumul{\overleftarrow{T}^{*}(\widehat{I}_1\widehat{I}_2\widehat{I}_3)} = \frac{V}{2\pi}T(1-T)(1-2T)\,,
\end{equation*}
for both coincident and non-coincident positions of scatterer and detector. This is the counterpart of Eq.~(\ref{eq_414}) for an energy--independent transparency $T(E)\equiv T$ and $V=\mu_\textsc{l}-\mu_\textsc{r}$. As mentioned in the introduction, contact terms make a non--vanishing contribution in the case $l=0$. More precisely, we have
\begin{align*}
	\lim\limits_{t\to +\infty}\frac{1}{t}\integral{0}{t}{^{3}t}
	\cumul{\overleftarrow{T}(\widehat{I}_1 \widehat{I}_2 \widehat{I}_3)}&=\frac{V}{2\pi}(-2T^2)(1-T)
		&\textnormal{(main term)}\,,\nonumber\\
	\lim\limits_{t\to +\infty}\frac{3}{t}\,\integral{0}{t}{^{2}t}
	\cumul{\overleftarrow{T}(\widehat{I}_1[\widehat{Q}_2,\widehat{I}_2])}&=0
		&\textnormal{(1st contact term)}\,,\nonumber\\
	\lim\limits_{t\to +\infty}\frac{1}{t}
	\integral{0}{t}{t_1}\cumul{[\widehat{Q}_1,[\widehat{Q}_1,\widehat{I}_1]]}&=\frac{V}{2\pi}T(1-T)
	 &\textnormal{(2nd contact term)}\,.
\end{align*}
In the second case, the main term alone contributes. In both cases a noteworthy feature is that the contribution of the main term remains unchanged if the time ordering is reduced as in $\widehat{I}_1\overleftarrow{T}(\widehat{I}_2 \widehat{I}_3)$ or omitted altogether.\\

\noindent{\bf 3rd cumulant of the second kind.} As explained in connection with Eq.~(\ref{eq_525}) the only missing item in order to establish binomial statistics is
\begin{equation*}
	\lim\limits_{t\to +\infty}\frac{1}{t}\integral{0}{t}{^{3}t}
	\cumul{\widehat{I}_1\overleftarrow{T}^{*}(\widehat{I}_2 \widehat{I}_3)}
	=\frac{V}{2\pi}T(1-T)(1-2T)\,.
\end{equation*}
Once again, the cumulant consists of a main and a contact term, see Eqs.~(\ref{eq_527}). The contribution of the main term is known, being insensitive to time ordering; hence we merely need to show that the contact term yields the complementary contribution. We distinguish between two cases:\\

(i) {\it Coincident positions, $l=0$:} From the computation of the 1st contact term in \cite{graf:09} we infer
\begin{align*}
	\lim\limits_{t\to +\infty}\frac{1}{t}\integral{0}{t}{^{2}t}
	\cumul{\widehat{I}_1 [\widehat{Q}_2,\widehat{I}_2])}
	&=-\frac{2i}{(2\pi)^2}T(1-T)\integral{}{}{x}\frac{\sin(Vx)}{(x-i\,0)^2}
	=\frac{V}{2\pi}T(1-T)\,,
\end{align*}
as claimed. In the second equality we used the fact that, as distributions, the odd part of $(x-i\,0)^{-2}$ is $((x-i\,0)^{-2}-(x+i\,0)^{-2})/2 = \pm i\pi\delta'(x)$.

(ii) {\it Separated positions, $l>0$:} As mentioned in Sect.~\ref{subsec_linear_model}, any contact term vanishes as a consequence of $[Q_l,I_l]=0$.

\section{Comparison between different approaches}
\label{comp}

There are ways and means to compute cumulants of transported charge, and in particular those of \cite{lesovik:03}, which at first sight differ from the ones used here on various counts. The purpose of this section is to show that they nevertheless are fundamentally the same.

A first difference rests on the use of Eq.~(\ref{eq_911}), as opposed to Eq.~(\ref{eq_919}) as above. The link is provided by identity~(\ref{eq_922}), which deserves proof anyway. We recall once more that its two sides are to be understood as power series in $\lambda$ with the time ordering placed inside the multiple integrals. The l.h.s. equals
	\begin{align*}
e^{i Ht}e^{-iH(\lambda)t}&=e^{i \lambda Q(t)}e^{-i \lambda Q}=\overleftarrow{T}\exp\bigl(i(Q(t)-Q)\bigr)
=\sum_{n=0}^\infty\frac{(i\lambda)^n}{n!}\overleftarrow{T}(Q(t)-Q))^n\,.
\end{align*}
In order to end up with things placed as stated, we apply the fundamental theorem of calculus to $\overleftarrow{T}(Q(t)-Q))^n$, rather than to $(Q(t)-Q))^n$:
\begin{align*}
\overleftarrow{T}(Q(t)-Q))^n&= \integral{0}{t}{t_1}
\frac{\partial}{\partial t_1}\left.\overleftarrow{T}\bigl((Q(t_1)-Q)\ldots
(Q(t_n)-Q)\bigr)\right|_{t_2 = \ldots = t_n = t}\\
&= \int_{0}^{t}\mathrm{d}t_1\ldots \mathrm{d}t_n\,\frac{\partial}{\partial t_n}\ldots\frac{\partial}{\partial t_1}\overleftarrow{T}\bigl((Q(t_1)-Q)\cdots (Q(t_n)-Q)\bigr)\,.
\end{align*}
We then expand the correlator in $Q(t_i)$ and $-Q$; the resulting second term is independent of $t_i$ and does not contribute to the derivative. By~(\ref{eq_921}) this proves~(\ref{eq_922}). A further proof of that identity, also given in \cite{graf:09}, is by comparing its two sides at each order $\lambda^n$. On the l.h.s. it is to be noted that
\begin{equation}\label{HI}
\HI(\lambda)=- \lambda I - i \frac{\lambda^{2}}{2}[Q,I] + O(\lambda^{3})
\end{equation}
is not homogeneous in $\lambda$; whereas the r.h.s. is expanded into $\overleftarrow{T}$-ordered products plus contact terms. It is instructive to check the case $n=2$. Up to a common factor $-\lambda^2/2$ the two sides are
\begin{equation*}
\integraal{0}{t}{t_1}{t_2}\overleftarrow{T}\bigl(I(t_1)I(t_2)\bigr)+\integral{0}{t}{t_1}[Q(t_1),I(t_1)]\,,\qquad
\integraal{0}{t}{t_1}{t_2}\overleftarrow{T}^{*}(I(t_1)I(t_2))\,;
\end{equation*}
by (\ref{ect}) they agree.

A further point deserving attention is as follows. For the model with quadratic dispersion the expansion (\ref{HI}) terminates at order $\lambda^{2}$ and reads in first quantization
\begin{equation*}
\HI(\lambda)=- \lambda (p Q'(x) + Q'(x) p) +\lambda^{2}Q'(x)^2\,,
\end{equation*}
see Eqs.~(\ref{eq_924}, \ref{eq_925}); and with
$\widehat{A}=\int\mathrm{d}x\mathrm{d}x'\widehat{\psi}(x)^*A(x,x')\widehat{\psi}(x')$ in second quantization
%
\begin{align}
\widehat{\HI}(\lambda)&=- \lambda \widehat{I} - i \frac{\lambda^{2}}{2}[\widehat{Q},\widehat{I}]\label{hatHI1}\\
&= \int \mathrm{d}x\bigl(- \lambda\widehat{\jmath}(x)Q'(x)+\lambda^2\widehat{\rho}(x)Q'(x)^2\bigr)\,,\label{hatHI2}
\end{align}
where charge and current densities are expressed in terms of fermionic operators $\widehat{\psi}(x)$ as
\begin{equation}\label{rhoj}
\widehat{\rho}(x)=\widehat{\psi}(x)^*\widehat{\psi}(x)\,,\qquad
\widehat{\jmath}(x)=-i\bigl(\widehat{\psi}(x)^*\widehat{\psi}'(x)-\widehat{\psi}'(x)^*\widehat{\psi}(x)\bigr)\,.
\end{equation}
Eq.~(\ref{hatHI2}) follows from (\ref{hatHI1}) by the commutation relation $[\widehat{Q},\widehat{I}]=\widehat{[Q,I]}$, but may also been obtained from $\widehat{H}(\lambda)$ as a starting point. Computations like those seen in Sect.~\ref{sec_derivations} can also be performed on the basis of Eqs.~(\ref{eq_911}, \ref{hatHI2}). However, a fact which was crucial there, but may escape notice here, is that the term of order $\lambda^2$ in (\ref{hatHI2}) is a commutator, as seen in (\ref{hatHI1}). The commutation relation, which states
\begin{equation*}
i\bigl[\int\mathrm{d}x\widehat{\rho}(x)Q(x), \int\mathrm{d}x'\widehat{\jmath}(x')Q'(x')\bigr]=-\int\mathrm{d}x\widehat{\rho}(x)Q'(x)^2\,,
\end{equation*}
has a seemingly independent derivation by means of
\begin{equation*}
i[\widehat{\rho}(x),\widehat{\jmath}(x')]=-\delta'(x'-x)\widehat{\rho}(x')\,,
\end{equation*}
which in turn follows from (\ref{rhoj}) and from
$[A^*B, C^*D]=A^*\{B, C^*\}D-C^*\{A^*,D\}B$ for annihilation operators $A$ through $D$.

\appendix
\appendixpage
\addappheadtotoc

\section{Some limits of distributions}
\label{sec_temp_distr}

In this appendix we collect some limits of distributions which are used in the main text. We consider distributions in one or two real variables, $x$ and $y$, and parametrized by $t>0$ (or $l>0$). In particular, we compute their limits as $t\to +\infty$. We recall the meaning of convergence for distributions: $D_t(x)\to D(x)$, $(t\to+\infty)$ stands for
\begin{equation*}
\lim_{t\to+\infty}\integral{\mathbb{R}}{}{x}D_t(x)\phi(x)
			=\integral{\mathbb{R}}{}{x}D(x)\phi(x)
\end{equation*}
for any test function $\phi$; say $\phi \in \mathcal{S}(\mathbb{R})$, the Schwartz space \cite{schwartz:66}.

\begin{proposition} We have
\label{prop_1}
\begin{align}
	\label{eq_32}
\integral{0}{t}{t_1}e^{ixt_1}=\frac{e^{ixt}-1}{ix}&\;\xrightarrow[t\to+\infty]{}\;\frac{i}{x+i\,0}\,,\\	
\label{eq_36}
\frac{1}{t}\iintegral{0}{t}{t_1}{0}{t_1}{t_2} e^{ixt_2}
=  \frac{1-e^{ixt}+ixt}{tx^{2}}&
		 \;\xrightarrow[t\to+\infty]{}\;  \frac{i}{x+i\,0} \,,\\
\label{eq_37}
\frac{1}{t}\iintegral{0}{t}{t_1}{0}{t}{t_2} e^{ix(t_1-t_2)}& \;\xrightarrow[t\to+\infty]{}\;2\pi \delta(x) \,,\\
\label{eq_38}
 \frac{1}{t}\iiintegral{0}{t}{t_{1}}{0}{t_{1}}{t_{2}}{0}{t_{2}}{t_{3}}
		e^{ixt_{1}} e^{iyt_{2}} e^{-i(x+y)t_{3}}& \;\xrightarrow[t\to+\infty]{}\; - \frac{1}{(x+i\,0)(x+y+i\,0)}\,,\\
\label{eq_44}
	\frac{1}{t} \left(\integral{0}{t}{t_1} e^{ixt_1}\right)
				\left(\iintegral{0}{t}{t_2}{0}{t_2}{t_3} e^{iy t_2}e^{-i(x+y)t_3}\right) &\;\xrightarrow[t\to+\infty]{}\; 2\pi i\delta(x)\frac{1}{y+i\,0}	
\,.
	\end{align}
\end{proposition}

\begin{proof}[Proof of (\ref{eq_32}).] We have
\begin{equation*}
\lim_{t\to+\infty}
\integral{0}{t}{t_1}e^{ixt_1}=\lim_{\substack{t\to+\infty\\ \varepsilon\downarrow 0}}\integral{0}{t}{t_1}e^{ixt_1}e^{-\varepsilon t_1}=
\lim_{\varepsilon\downarrow 0}\integral{0}{\infty}{t_1}e^{ixt_1}e^{-\varepsilon t_1}\,.
\end{equation*}
The computation is justified, because the limits $t\to+\infty$ and $\varepsilon\downarrow 0$ interchange. This is seen by testing the middle expression with $\phi$: For either order of limits one obtains $\integral{0}{\infty}{t_1}\hat{\phi}(-t_1)$, where $\hat{\phi}(-t_1)=\integral{\mathbb{R}}{}{x}e^{ixt_1}\phi(x)$. Finally, the expression under the limit on the r.h.s. equals $i(x+i\varepsilon)^{-1}$.\end{proof}

\begin{proof}[Proof of (\ref{eq_36}).] As a general fact, convergence implies in the Ces\`aro sense; specifically $\lim_{t\to+\infty}D_t= D$ implies $\lim_{t\to+\infty}t^{-1}\integral{0}{t}{t_1}D_{t_1}=D$. By applying this remark to (\ref{eq_32}) the claim follows. \end{proof}

\begin{proof}[Proof of (\ref{eq_37}).] The integral factorizes and equals by Eq.~(\ref{eq_32})
\begin{equation*}
\frac{1}{t}\frac{e^{ixt}-1}{ix}\frac{e^{-ixt}-1}{-ix}=\frac{2-e^{ixt}-e^{-ixt}}{tx^{2}}\;\xrightarrow[t\to+\infty]{}\;2\pi  \delta(x)\,.
\end{equation*}
The limit was established by observing that the expression on the l.h.s. is (twice) the real part of that in (\ref{eq_36}). On its r.h.s. we used the SW formula (\ref{ws}).
\end{proof}
\begin{proof}[Proof of (\ref{eq_38}).] By the remark made earlier the average $t^{-1}\integral{0}{t}{t_1}$ may be omitted from its l.h.s. and the limit replaced by $t_1\to+\infty$. The remaining double integral is computed by means of
\begin{equation*}
\iintegral{0}{t_1}{t_2}{0}{t_2}{t_3} e^{iy t_2}e^{-i(x+y)t_3}
=\frac{1-e^{-ixt_1}}{x(x+y)}+\frac{1-e^{iyt_1}}{y(x+y)}
\end{equation*}
as
\begin{equation*}
\frac{e^{ixt_{1}}-1}{x\,(x+y)} - \frac{e^{i(x+y)t_{1}}-e^{ixt_{1}}}{y\,(x+y)}  = \frac{1}{y}\left( \frac{e^{ixt_{1}}-1}{x} - \frac{e^{i(x+y)t_{1}}-1}{x+y}\right)\,.
\end{equation*}
The prefactor on the r.h.s., $y^{-1}$, may be given a regularization, $(y\pm i\,0)^{-1}$. It is dispensable, because the expression inside the brackets vanishes for $y=0$, but it allows to remove them. More precisely, either term so obtained is a well-defined distribution; in fact a product distribution in $x$ and $y$, respectively in $x+y$ and $y$. By (\ref{eq_32}) the sought limit is
\begin{equation*}
-\frac{1}{y\pm i\,0}\left( \frac{1}{x+i\,0} -\frac{1}{x+y+i\,0} \right) = -\frac{1}{(x+i\,0)(x+y+i\,0)}\,.
\end{equation*}
\end{proof}

\begin{proof}[Proof of (\ref{eq_44}).] The integrals on the l.h.s of (\ref{eq_38}, \ref{eq_44}) only differ by their domains. In fact, while in the first expression we demand $t_1\geqslant t_2\geqslant t_3$, in the second only $t_2\geqslant t_3$ is required. This is however equivalent to adding the cases $t_2\geqslant t_1\geqslant t_3$ and $t_2\geqslant t_3\geqslant t_1$ to the previous one. By (\ref{eq_38}) this yields in the limit
\begin{equation*}
	-\frac{1}{(x+i\,0)(x+y+i\,0)}-\frac{1}{(y+i\,0)(x+y+i\,0)}-\frac{1}{(y+i\,0)(-x+i\,0)}
	= 2\pi i\,\delta(x)\frac{1}{y+i\,0}\,,
\end{equation*}
where we applied the SW formula (\ref{ws}) to the third term on the l.h.s.
\end{proof}

\begin{proposition} For any test functions $\rho$, $\rho_i$, $(i=1,2,3)$ we have 	
\begin{align}
	\label{eq_41}
	\frac{\rho(lx)}{x\pm i\,0}
	\;\xrightarrow[l\to+\infty]{}&\; \left( \integral{\mathbb{R}}{}{v} \frac{\rho (v)}{v \pm i\,0}  \right) \delta(x)\,,\\
\label{eq_42}
\frac{\rho_1(lx)\rho_2(ly)\rho_3(l(x-y))}{(x\pm i\,0)(y\pm i\,0)}
	&\;\xrightarrow[l\to+\infty]{}\; \left( \integraal{\mathbb{R}^2}{}{u}{v} \frac{\rho_1(u)\rho_2(v)\rho_3(u-v)}{(u \pm i\,0)(v \pm i\,0)}  \right) \delta(x) \delta(y)
\,.
	\end{align}
\end{proposition}
\begin{proof} [Proof of (\ref{eq_41}).] We probe the convergence on any test function $\phi$. By substituting $x$ with $x/l$ the claim may be recast as
	\begin{equation*}
	\integral{\mathbb{R}}{}{x}
		\frac{\rho(x)}{x\pm i\,0} \bigl( \phi(x/l) - \phi(0)\bigr) \;\xrightarrow[l\to+\infty]{}\; 0\,.	\end{equation*}
The mean value theorem yields $|\phi(x) - \phi(0)|\le\Vert\phi'\Vert_\infty|x|$ and the sufficient bound
	\begin{equation*}
\integral{\mathbb{R}}{}{x}\Bigl|\frac{\rho(x)}{x}\Bigr|\Vert\phi'\Vert_\infty\Bigl|\frac{x}{l}\Bigr|=\Vert\rho\Vert_1\Vert\phi'\Vert_\infty l^{-1}\,.
		\end{equation*}	
\end{proof}
\begin{proof} [Proof of (\ref{eq_42}).] The proof is similar to the previous one. To be shown is
\begin{equation*}
	\integraal{\mathbb{R}^2}{}{x}{y}\frac{\rho_1(x)\rho_2(y)\rho_3(x-y)}{(x\pm i\,0)(y\pm i\,0)} \bigl( \phi(x/l,y/l) - \phi(0,0)\bigr) \;\xrightarrow[l\to+\infty]{}\; 0\,.	
\end{equation*}
However, instead of using the mean value theorem, we may write
\begin{equation*}
\phi(x,y) - \phi(0,0)=x\phi_1(x)+y\phi_2(y)+xy\phi_{12}(x,y)\,,
\end{equation*}
for some further test functions $\phi_i$, $(i=1,2,12)$. Three terms then require estimation. The first one is $l^{-1}$ times
\begin{equation*}
\integral{\mathbb{R}}{}{x}\rho_1(x)\rho(x)\phi_1(x/l)=O(1)\,,
\end{equation*}
where $\rho(x)=\integral{\mathbb{R}}{}{y}(y\pm i\,0)^{-1}\rho_2(y)\rho_3(x-y)$ is a bounded function; likewise for the second term. The third one is $O(l^{-2})$.
\end{proof}

\section{Matrix elements of current and charge}
\label{sec_matrix_elements}

In this appendix we will derive the matrix elements of current $I$ and charge $Q$ in the quadratic dispersion model of Sect.~\ref{subsec_quadratic_model}. We also obtain a splitting $\matel{1}{I}{2} = \matel{1}{I}{2}^{(+)}+\matel{1}{I}{2}^{(-)}$ suitable for Eq.~(\ref{eq_427}), i.e.
\begin{equation*}
	i\matel{1}{Q}{2} = \frac{\matel{1}{I}{2}^{(+)}}{E_1-E_2+i\,0} + \frac{\matel{1}{I}{2}^{(-)}}
		{E_1-E_2-i\,0}\,,
\end{equation*}
where $\ket{i}$, ($i=1,2$) denotes the Lippmann-Schwinger state and $E_i=k_i^2$.

We shall first compute $\matel{1}{Q}{2}$. Since $Q(x)$ is supported in $x > x_0$, the wave-function $\psi_{k}(x)$ of $\ket{i}$ will matter only in that region, see Eq.~(\ref{eq_21}). Next we will deduce $\matel{1}{I}{2}$ from Eqs.~(\ref{eq_425a}, \ref{eq_425}), i.e. from
\begin{gather*}
\matel{1}{I}{2}=i(E_1-E_2)\matel{1}{Q}{2}\,,\qquad
	\widehat{Q}(k) = (-i)\frac{\widehat{Q'}(k)}{k-i\,0}\,,
\end{gather*}
where $k=\pm k_1\pm k_2$, with signs occurring in various combinations. Finally the splitting is obtained by factoring $k-i\,0$ out of $E_1-E_2\pm i\,0$, with a sign depending on cases.
\begin{itemize}
	\item $k_1,k_2>0$.
	\begin{align}
		\matel{1}{Q}{2} &= \integral{-\infty}{\infty}{x} \overline{\psi_{k_1}(x)}Q(x)\psi_{k_2}(x)= \integral{-\infty}{\infty}{x} \overline{t(k_1)} e^{-ik_1 x} Q(x) t(k_2) e^{ik_2 x}\nonumber\\
&= \overline{t(k_1)} t(k_2) \widehat{Q}(k_1 - k_2)\,.\label{eq_B2}
	\end{align}
The factorization is $E_1-E_2-i\,0 =(k_1+k_2)(k_1-k_2-i\,0)$, since $k_1+k_2>0$ in the case considered. It is used twice: Without regularization, to obtain
 \begin{equation}
 \label{eq_B4}
		\matel{1}{I}{2} = (k_1 + k_2) \overline{t(k_1)}t(k_2) \widehat{Q'}(k_1-k_2)\,;
\end{equation}
and with it, to get $\matel{1}{I}{2}^{(-)}=\matel{1}{I}{2}$ and $\matel{1}{I}{2}^{(+)}=0$. In particular, by $\widehat{Q'}(0)=1$, we obtain Eq.~(\ref{eq_424}).
	\item $k_1>0$, $k_2<0$.
	\begin{align*}
		\matel{1}{Q}{2}& = \integral{-\infty}{\infty}{x} \overline{t(k_1)} e^{-ik_1 x} Q(x) (e^{ik_2 x} + r(k_2) e^{-ik_2 x})\\
&= \overline{t(k_1)} \widehat{Q}(k_1-k_2) + \overline{t(k_1)}r(k_2)\widehat{Q}(k_1+k_2)\,.
	\end{align*}
The factorization for the second term is $E_1-E_2-i\,0 =(k_1-k_2)(k_1+k_2-i\,0)$; but remains arbitrary for the first one, $E_1-E_2\pm i\,0=(k_1+k_2)(k_1-k_2-i\,0)$, since $k_1-k_2\neq 0$ here. Hence we have
 \begin{equation}
\label{me+-}
		\matel{1}{I}{2}= (k_1+k_2)\overline{t(k_1)}\widehat{Q'}(k_1-k_2) + (k_1-k_2)\overline{t(k_1)}r(k_2)\widehat{Q'}(k_1+k_2)\,,
\end{equation}
and we may again set $\matel{1}{I}{2}^{(-)}=\matel{1}{I}{2}$.
	\item $k_1<0$, $k_2>0$. This case follows from the previous one by complex conjugation and by interchanging $k_1$ and $k_2$:
	\begin{align}
		\matel{1}{Q}{2} &= t(k_2) \widehat{Q}(k_1-k_2) + t(k_1)\overline{r(k_1)}\widehat{Q}(-k_1-k_2)\,,\label{meq+-}\\
			\matel{1}{I}{2} &= (k_1+k_2)t(k_2)\widehat{Q'}(k_1-k_2) + (k_2-k_1)t(k_2)\overline{r(k_1)}\widehat{Q'}(-k_1-k_2)\,.
\label{me-+}	\end{align}
We still have $\matel{1}{I}{2}^{(-)}=\matel{1}{I}{2}$, since $-i(E_1-E_2-i\,0)$ is invariant under the move.
	\item $k_1,k_2<0$. We have
	\begin{align}
		\matel{1}{Q}{2} &= \integral{-\infty}{\infty}{x} (e^{-ik_1 x} + \overline{r(k_1)} e^{ik_1 x}) Q(x)
				(e^{ik_2 x} + r(k_2) e^{-ik_2 x})\nonumber\\
			&= \widehat{Q}(k_1-k_2) + \overline{r(k_1)} \widehat{Q}(-k_1-k_2) + r(k_2) \widehat{Q}(k_1+k_2)
				+ \overline{r(k_1)}r(k_2)\widehat{Q}(k_2-k_1)\,.\label{eq_B3}
\end{align}
The arguments of the two middle terms, $\mp(k_1+k_2)$, are non-vanishing in the case considered, hence their regularization by $-i\,0$ irrelevant. They may thus be linked to either term $\matel{1}{I}{2}^{(\pm)}$. Using $(k_1+k_2)(k_1-k_2\mp i\,0)=E_1-E_2\pm i\,0$ on the remaining two terms we have
\begin{align}
	\matel{1}{I}{2}
		=& (k_1+k_2)\widehat{Q'}(k_1-k_2) + (k_2-k_1)\widehat{Q'}(-k_1-k_2)\nonumber \\
&+(k_1-k_2)r(k_2) \widehat{Q'}(k_1+k_2)-(k_1+k_2) \overline{r(k_1)}r(k_2)\widehat{Q'}(k_2-k_1)\label{me--}
\end{align}
and we may choose to set
\begin{equation*}
	\matel{1}{I}{2}^{(-)} = -(k_1+k_2)\overline{r(k_1)}r(k_2)\widehat{Q'}(k_2-k_1)\,.
\end{equation*}
In particular, we find Eq.~(\ref{eq_431}).
\end{itemize}

\noindent
{\bf Acknowledgments.} We thank M. Reznikov for discussions and encouragement.

\bibliography{to_and_cs_biblio}
\bibliographystyle{unsrt}
\end{document}